\documentclass[runningheads]{llncs}
\usepackage{graphicx}
\usepackage{hyperref}
\usepackage{amsmath}
\usepackage{cleveref}
\usepackage{algorithm}
\usepackage{algorithmic}
\usepackage{cite}
\usepackage{tikz}
\usepackage{appendix}
\usepackage{verbatim}

\usepackage[utf8]{inputenc}

\newcommand{\cost}{\ensuremath\textrm{cost}}
\begin{document}
\title{Upper Dominating Set: Tight Algorithms for Pathwidth and Sub-Exponential Approximation}

\titlerunning{Upper Dominating Set: Tight Algorithms}
%
%
\author{Louis Dublois\inst{1} 
\and
Michael Lampis\inst{1}
\and 
Vangelis Th. Paschos\inst{1}}

%
\authorrunning{L. Dublois, M. Lampis, V. Th. Paschos}
%
\institute{Université Paris-Dauphine, PSL University, CNRS, LAMSADE, Paris, France
\email{\{louis.dublois,michail.lampis,paschos\}@lamsade.dauphine.fr}}
%
\maketitle              
\begin{abstract}
An upper dominating set is a minimal dominating set in a graph. In the \textsc{Upper Dominating Set} problem, the goal is to find an upper dominating set of maximum size. We study the complexity of parameterized algorithms for \textsc{Upper Dominating Set}, as well as its sub-exponential approximation. First, we prove that, under ETH, \textsc{$k$-Upper Dominating Set} cannot be solved in time $O(n^{o(k)})$ (improving on $O(n^{o(\sqrt{k})})$), and in the same time we show under the same complexity assumption that for any constant ratio $r$ and any $\varepsilon > 0$, there is no $r$-approximation algorithm running in time $O(n^{k^{1-\varepsilon}})$. Then, we settle the problem's complexity parameterized by pathwidth by giving an algorithm running in time $O^*(6^{pw})$ (improving the current best $O^*(7^{pw})$), and a lower bound showing that our algorithm is the best we can get under the SETH. Furthermore, we obtain a simple sub-exponential approximation algorithm for this problem: an algorithm that produces an $r$-approximation in time $n^{O(n/r)}$, for any desired approximation ratio $r < n$. We finally show that this time-approximation trade-off is tight, up to an arbitrarily small constant in the second exponent: under the randomized ETH, and for any ratio $r > 1$ and $\varepsilon > 0$, no algorithm can output an $r$-approximation in time $n^{(n/r)^{1-\varepsilon}}$. Hence, we completely characterize the approximability of the problem in sub-exponential time. 
\end{abstract}

%
%
%
\section{Introduction}\label{section:introduction}

In a graph $G = (V,E)$, a set $D \subseteq V$ is called a \textit{dominating set} if all vertices of $V$ are dominated by $D$, that is for every $u \in V$ either $u$ belongs to $D$ or $u$ is a neighbor of some vertex in $D$. The well-known \textsc{Dominating Set} problem is studied with a minimization objective: given a graph, we are interested in finding the smallest dominating set. In this paper, we consider \textit{upper dominating sets}, that is dominating sets that are minimal, where a dominating set $D$ is minimal if no proper subset of it is a dominating set, that is if it does not contain any redundant vertex. We study the problem of finding an upper dominating set of maximum size. 

This problem is called \textsc{Upper Dominating Set}, and is the Max-Min version of the \textsc{Dominating Set} problem. We call \textsc{Upper Dominating Set} the considered optimization problem and \textsc{$k$-Upper Dominating Set} the associated decision problem. 

Studying Max-Min and Min-Max versions of some famous optimization problems is not a new idea, and it has recently attracted some interest in the literature: \textsc{Minimum Maximal Independent Set} \cite{BourgeoisCEP13, Halldorsson93a, HurinkN08} (also known as \textsc{Minimum Independent Dominating Set}), \textsc{Maximum Minimal Vertex Cover}  \cite{BoriaCP13, Zehavi15}, \textsc{Maximum Minimal Separator} \cite{HanakaBZO19}, \textsc{Maximum Minimal Cut} \cite{EtoHK019}, \textsc{Minimum Maximal Knapsack} \cite{ArkinBMS03, FuriniLS17, GourvesMP13} (also known as \textsc{Lazy Bureaucrat Problem}), \textsc{Maximum Minimal Feedback Vertex Set} \cite{DubloisHKLM20}. In fact, the original motivation for studying these problems was to analyze the performance of naive heuristics compared to the natural Max and Min versions, but these Max-Min and Min-Max problems have gradually revealed some surprising combinatorial structures, which makes them as interesting as their natural Max and Min versions. The \textsc{Upper Dominating Set} problem can be seen as a member of this framework, and studying it within this framework is one of our motivation. 

This problem is also one of the six problems of the well-known \textit{domination chain} (see \cite{Haynes0093827, bazgan2019domination}) and is somewhat one which has fewer results, compared to the famous \textsc{Dominating Set} and \textsc{Independent Set} problems. Increasing our understanding of the \textsc{Upper Dominating Set} problem compared to these two famous problems is another motivation. 

\textsc{Upper Dominating Set} was first considered in an algorithmic point of view by Cheston et al. \cite{ChestonFHJ90}, where they showed that the problem is NP-hard. In the more extensive paper considering this problem, Bazgan et al. \cite{BazganBCFJKLLMP18} studied approximability, and classical and parameterized complexity of the \textsc{Upper Dominating Set} problem. In the polynomial approximation paradigm, they proved that the problem does not admit an $n^{1-\varepsilon}$-approximation for any $\varepsilon > 0$, unless P=NP, making the problem as hard as \textsc{Independent Set}, whereas there exists a greedy $\ln n$-approximation algorithm for the Min version \textsc{Dominating Set}.

Considering the parameterized complexity, they proved that the problem is as hard as the \textsc{$k$-Independent Set} problem: \textsc{$k$-Upper Dominating Set} is W[1]-hard parameterized by the standard parameter $k$. Nonetheless, in their reduction, there is an inherent quadratic blow-up in the size of the solution $k$, so they essentially proved that there is no algorithm solving \textsc{$k$-Upper Dominating Set} in time $O(n
^{o(\sqrt{k})})$. They also gave FPT algorithms parameterized by the pathwidth $pw$ and the treewidth $tw$ of the graph, in time $O^*(7^{pw})$\footnote{$O^*$ notation suppresses polynomial factors in the input size.} and $O^*(10^{tw})$, respectively. 

\subsubsection*{Our results:} The state of the art summarized above motivates two basic questions: first, can we close the gap between the lower and upper bounds of the complexity of the problem parameterized by pathwidth ; second, since the polynomial approximation is essentially settled, can we design sub-exponential approximation algorithms which can reach any approximation ratio $r < n$ ? We answer these questions and along the way we give stronger FPT hardness results. In fact, we prove the following:

(i) In Section \ref{section:fpthardness}, we show the following: under ETH, there is no algorithm solving \textsc{$k$-Upper Dominating Set} in time $O(n^{o(k)})$ ; and under the same complexity assumption, for any ratio $r$ and any $\varepsilon > 0$, there is no algorithm for this problem that outputs an $r$-approximation in time $O(n^{k^{1-\varepsilon}})$. 

(ii) In Section \ref{section:pathwdith}, we give a dynamic programming algorithm parameterized by pathwidth that solves \textsc{Upper Dominating Set} in time $O^*(6^{pw})$. Surprisingly, this result is obtained by slightly modifying the algorithm of Bazgan et al. \cite{BazganBCFJKLLMP18}. We then prove the following: under SETH, and for any $\varepsilon > 0$, \textsc{Upper Dominating Set} cannot be solved in time $O^*((6-\varepsilon)^{pw})$. This is our main result, and it shows that our algorithm for pathwidth is optimal. 

(iii) In Section \ref{section:subexponential}, we give a simple time-approximation trade-off: for any ratio $r < n$, there exists an algorithm for \textsc{Upper Dominating Set} that ouputs an $r$-approximation in time $n^{O(n/r)}$. We also give a matching lower bound: under the randomized ETH, for any ratio $r > 1$ and any $\varepsilon > 0$, there is no algorithm that outputs an $r$-approximation running in time $n^{(n/r)^{1-\varepsilon}}$.

\section{Preliminaries}\label{section:preliminaries}

We use standard graph-theoretic notation and we assume familiarity with the basics of parameterized complexity (e.g. pathwidth, the SETH and FPT algorithms), as given in \cite{CyganFKLMPPS15}. Let $G = (V,E)$ be a graph with $|V| = n$ vertices and $|E| = m$ edges. For a vertex $u \in V$, the set $N(u)$ denotes the set of neighbors of $u$, $d(u) = |N(u)|$, and $N[u]$ the closed neighborhood of $u$, i.e. $N[u] = N(u) \cup \{ u \}$. For a subset $U \subseteq V$ and a vertex $u \in V$, we note $N_U(u) = N(u) \cap U$. Furthermore, for $U \subseteq V$, we note $N(U) = \bigcup_{u \in U} N(u)$. For an edge set $E' \subseteq E$, we use $V(E')$ to denote the set of its endpoints. For $V' \subseteq V$, we note $G[V']$ the subgraph of $G$ induced by $V'$. 

An \textit{upper dominating set} $D \subseteq V$ of a graph $G(V,E)$ is a set of vertices that dominates all vertices of $G$, and which is minimal. Note that $D$ is minimal if we have the following: for every vertex $u \in D$, either $u$ has a private neighbor, that is a neighbor that is dominated only by $u$, or $u$ is its own private vertex, that is $u$ is only dominated by itself. We note an upper dominating set $D = S \cup I$, where $S$ is the set of vertices of $D$ which have at least one private neighbor, and $I$ is the set of vertices of $D$ which forms an independent set, that is the set of vertices which are their own private vertices. 

Note that a maximal independent set $I$ (also known as an independent dominating set) is an upper dominating set since it is a set of vertices which dominates the whole graph and such that every vertex $u \in I$ is its own private vertex.

\section{FPT and FPT-approximation Hardness}\label{section:fpthardness}

In this section, we present two hardness results for the \textsc{$k$-Upper Dominating Set} problem in the parameterized paradigm: we prove first that the considered problem cannot be solved in time $O(n^{o(k)})$ under the ETH ; and we prove then under the same complexity assumption that for any constant approximation ratio $0 < r < 1$ and any $\varepsilon > 0$, there is no FPT algorithm giving an $r$-approximation for the \textsc{$k$-Upper Dominating Set} problem running in time $O(n^{k^{1-\varepsilon}})$. 


Note that \textsc{$k$-Upper Dominating Set} being W[1]-hard was already proved by Bazgan et al. \cite{BazganBCFJKLLMP18}. To get this result, they made a reduction from the \textsc{$k$-Multicolored Clique} problem. Nonetheless, in this reduction, the size of the solution of the \textsc{$k$-Upper Dominating Set} problem was quadratic compared to the size of the solution of the \textsc{$k$-Multicolored Clique} problem. Thus, they proved essentially the next result: \textsc{$k$-Upper Dominating Set} problem cannot be solved in time $O(n^{o(\sqrt{k})})$.

To obtain our desired negative results, we will make a reduction from the \textsc{$k$-Independent Set} problem to our problem. So recall that we have the following hardness results for the \textsc{$k$-Independent Set} problem:


\begin{lemma}[Theorem 5.5 from \cite{ChenHKX06}]\label{lem:w1hardis}
Under ETH, \textsc{$k$-Independent Set} cannot be solved in time $O(n^{o(k)})$. 
\end{lemma}

\begin{lemma}[Corollary 2 from \cite{BonnetE0P13}]\label{lem:fptapproxhardis}
Under ETH, for any constant $r > 0$ and any $\varepsilon > 0$, there is no $r$-approximation algorithm for \textsc{$k$-Independent Set} running in time $O(n^{k^{1-\varepsilon}})$.
\end{lemma}

We will obtain similar results for the \textsc{$k$-Upper Dominating Set} by doing a reduction from \textsc{$k$-Independent Set}. This reduction will linearly increase the size of the solutions between the two problems, so these two hardness results for the latter problem will hold for the former problem. 

Before we proceed further in the description of our reduction, note that we will use a variant of the \textsc{$k$-Independent Set} problem. In this variant, the graph $G$ contains $k$ cliques which are connected to each other, and if a solution of size $k$ exists, then this solution takes exactly one vertex per clique. Note that the Lemmas \ref{lem:w1hardis} and \ref{lem:fptapproxhardis} hold on this particular instance, since this is a case where the problem remains hard to solve in FPT time and to approximate in FPT time. So we will use this variant. 

Let us now present our reduction. We are given a \textsc{$k$-Independent Set} instance $G$ with $n$ vertices and $m$ edges, where the $n$ vertices are partitioned in $k$ distinct cliques $V_1, \ldots, V_k$ connected to each other. We define the following number: $A = 5$. We set our budget to be $k' = Ak$. 

We construct our instance $G'$ of \textsc{$k'$-Upper Dominating Set} as follows:

\begin{enumerate}
    \item For any vertex $u \in V(G)$, create an independent set $Z_u$ of size $A$. 
    \item For any edge $(u,v) \in E(G)$, add all edges between the vertices of $Z_u$ and the vertices of $Z_v$. 
    \item For any $i \in \{ 1, \ldots, n \}$, let $W_i$ be the \textit{group} associated to the clique $V_i$, which contains all vertices of all independent sets $Z_u$ such that the vertex $u$ belongs to the clique $V_i$. For any $i \in \{ 1, \ldots, k \}$, create a vertex $z_i$ connected to all vertices of the group $W_i$. 
\end{enumerate}

Now that we have presented our reduction, we argue that it is correct. Recall that the target size of an optimal solution in $G'$ is $k'$ as defined above. We can prove that, given an independent set $I$ of size at least $k$ in $G$, we can construct an upper dominating set of size at least $Ak$ in $G'$ by taking the $A$ vertices of the independent set $Z_u$ for any vertex $u \in I$.

\begin{lemma}\label{lem:forwardredisuds}
If $G$ has an independent set of size at least $k$, then $G'$ has an upper dominating set of size at least $k'$. 
\end{lemma}

\begin{proof}
Assume $G$ admits an independent set $I$ of size at least $k$. We construct an upper dominating set $D$ of size at least $k'$ in $G'$ as follows: for any vertex $u \in V(G) \cap I$, put the $A$ vertices of the corresponding independent set $Z_u$ in $D$.

Clearly, the size of $D$ is at least $k' = Ak$ since $I$ is of size at least $k$ and every independent set $Z_u$ is of size $A$. 

Consider any $i \in \{ 1, \ldots, k \}$, and observe that $I \cap V_i \leq 1$ since $V_i$ is a clique. But there is $k$ such cliques $V_i$ and $|I| \geq k$, so necessarily $|I| = k$ and $|I \cap V_i| = 1$ for all $i \in \{ 1, \ldots, k \}$. 

Consider again any $i \in \{ 1, \ldots, k \}$ and the corresponding vertex $u$ which belongs to $I \cap V_i$. The vertices of the independent set $Z_u$ have been put in $D$. By the construction, the vertices of $Z_u$ are connected to all remaining vertices of $W_i$ and to the vertex $z_i$. So the vertices of $W_i \cup \{ z_i \}$ are dominated. This is true for all $i \in \{ 1, \ldots, k \}$, so the graph $G'$ is dominated by $D$. 

Moreover, since $I$ is an independent set, and by the construction, it follows that, for any two vertices $u, u' \in I$, there is no edge between the vertices of $Z_u$ and the vertices of $Z_{u'}$. And since the sets $Z_u$ are independent sets, it follows that $D$ is an independent set of $G'$, which means that all vertices of $D$ are their own private vertices.
\qed
\end{proof}

The idea of the following proof is the following: if an upper dominating set in $G'$ of size at least $k'$ has not the form described in Lemma \ref{lem:forwardredisuds}, then it cannot have size at least $k'$, enabling us to construct an independent set of size at least $k$ in $G$ from an upper dominating set which has the desired form. 

\begin{lemma}\label{lem:backwardredisuds}
If $G'$ has an upper dominating set of size at least $k'$, then $G$ has an independent set of size at least $k$. 
\end{lemma}

\begin{proof}
Assume $G'$ admits an upper dominating set $D$ of size at least $k'$. For any $i \in \{ 1, \ldots, k \}$, we give the following notations:

\begin{itemize}
    \item If there exists at least three vertices $u_1, u_2, u_3$ in $V_i$ such that $Z_{u_j} \cap D \neq \emptyset$ for all $j \in \{ 1, 2, 3 \}$, then we call the group $W_i$ \textit{very bad}. 
    \item If there exists exactly two vertices $u$ and $u'$ in $V_i$ such that $Z_u \cap D \neq \emptyset$ and $Z_{u'} \cap D \neq \emptyset$, then we call the group $W_i$ \textit{bad}. 
    \item Otherwise, that is if there exists at most one vertex $u \in V_i$ such that $Z_u \cap D \neq \emptyset$, then $W_i$ is called \textit{good}. 
\end{itemize}

Our proof will therefore be to consider that there is some bad and very bad groups in $G'$ and we will arrive at a contradiction on the size of $D$, which will prove that the solution $D$ as the form described in Lemma \ref{lem:forwardredisuds}. 

So suppose there exists a bad or very bad group $W_i$. Suppose first that there exists $u \in V_i$ such that $|Z_u \cap D| \geq 2$. Observe that the vertices of $Z_u$ which belong to $D$ have the same neighborhood (since they are in the same independent set $Z_u$), and observe that there exists at least one vertex $u' \in V_i \setminus \{ u \}$ with $Z_{u'} \cap D \neq \emptyset$ (since $W_i$ is bad or very bad). So the vertices of $Z_u$ which belong to $D$ are dominated and share the same neighborhood, which contradicts the fact that $D$ is an upper dominating set. So, for a bad or very bad group $W_i$ and any $u \in V_i$ such that $Z_u \cap D \neq \emptyset$, we have $|Z_u \cap D| = 1$. 

Consider now a bad group $W_i$ and the two vertices $u, u' \in V_i$ such that $Z_u \cap D \neq \emptyset$ and $Z_{u'} \cap D \neq \emptyset$. Let $v = Z_u \cap D$ and $v' = Z_{u'} \cap D$. The two vertices $v$ and $v$' dominate each other since they belong to two distinct independent set $Z_u$ and $Z_{u'}$ in the group $W_i$. Observe now that $z_i$ cannot be in $D$ because otherwise it is dominated and have no private neighbor (since all vertices of $W_i$ are dominated by $v$ or $v'$). Moreover, since $W_i$ is a bad group and $|Z_u \cap D| = 1$ and $|Z_{u'} \cap D| = 1$, we have $|(W_i \cup \{ z_i \}) \cap D| = 2$. 

Consider now a very bad group $W_i$ and three vertices $u_1, u_2, u_3$ in $V_i$ such that $Z_{u_j} \cap D \neq \emptyset$ for all $j \in \{ 1, 2, 3 \}$. Let $v_j = Z_{u_j} \cap D$ (since $|Z_{u_j} \cap D| = 1$). Consider now any $j \in \{ 1, 2, 3 \}$. Observe that $W_i \cup \{ z_i \}$ is dominated by the two vertices $v_{j'}$ and $v_{j''}$, for $j', j'' \in \{ 1, 2, 3 \}$ and $j', j'' \neq j$. Indeed, $(W_i \cup \{ z_i \}) \setminus Z_{u_{j'}}$ is dominated by $v_{j'}$ and $Z_{u_{j'}}$ is dominated by $v_{j''}$. Since $v_j$ is dominated by both $v_{j'}$ and $v_{j''}$, it follows that $v_j$ necessarily have a private neighbor outside $W_i \cup \{ z_i \}$. It is true for any vertex $u \in V_i$ such that $|Z_u \cap D| \neq \emptyset$.

Now consider such a vertex $v \in \{ v_1, v_2, v_3 \}$ and his private neighbor $w$, which let say is in $W_{i'}$, for $i' \in \{ 1, \ldots, k \} \setminus \{ i \}$, and in $Z_{u_w}$, for $u_w \in V_{i'}$. Since $w$ is a private neighbor of $v$, it follows that, for any $u' \in V_{i'} \setminus \{ u_w \}$, no vertex of $Z_{u'}$ is in $D$, because otherwise $w$ would not be a private neighbor of $v$. Moreover, the vertex $z_{i'}$ cannot be in $D$ by the same argument. So necessarily, the group $W_{i'}$ is a good group. 

Now suppose that there exists at least two vertices $w_1, w_2 \in (Z_{u_w} \setminus \{ w \}) \cap D$. Observe that both $w_1$ and $w_2$ are dominated by $v$, since there is all edges between the vertices of $Z_u$ (where $v$ belongs) and the vertices of $Z_{u_w}$. But $w_1$ and $w_2$ have the same neighborhood, which contradicts the fact that $D$ is an upper dominating set. So $|Z_{u_w} \cap D| \leq 1$. 

Consider any $i \in \{ 1, \ldots, k \}$ and observe that the vertex $z_i$ has to be dominated, and its neighborhood is the group $W_i$, so $|(W_i \cup \{ z_i \}) \cap D| \geq 1$. 

Now reconsider the vertex $v$. Since $W_{i'}$ is a good group, since $z_i$ does not belong to $D$, and since $|Z_{u_w} \cap D| \leq 1$, we obtain $|Z_{u_w} \cap D| = 1$. 

Now consider two vertices $v$ and $v'$ belonging to $D$ and which are in the same very bad group or in two distinct very bad groups. And consider their corresponding private neighbors $w$ and $w'$. Clearly, $w$ and $w'$ do not belong to the same independent set $Z_{u_w}$, since $v$ dominates all vertices of this independent set. Moreover, since $|Z_{u_w} \cap D| = 1$ for the independent set $Z_{u_w}$ which contains the vertex $w$, the two vertices $w$ and $w'$ cannot belong to the same good group since the vertex in $Z_{u_w} \cap D$ dominates all remaining vertices of $W_{i'} \cup \{ z_{i'} \}$. So, for any two such vertices $v$ and $v'$ belonging to $D$ (whether they are in the same very bad group or not), their corresponding private neighbors are in distinct good groups. 

But, since in these good groups (which contain a private neighbor $w$) we have $|W_{i'} \cap D| = 1$, we have the following: for any vertex $v \in D$ which belongs to a very bad group, there exists at least one distinct good group $W_{i'}$ where a single vertex is in $D$. 

Now, let $b$ be the number of bad groups and $B$ be the number of vertices in the very bad groups. 

Observe now that, in a good group $W_i$ which does not contain a private neighbor $w$ of a vertex $v$ belonging to a very bad group, we have that at most $A$ vertices of $W_i \cup \{ z_i \}$ are in $D$: since $W_i$ is a good group, there exists at most one vertex $u \in V_i$ such that $Z_u \cap D \neq \emptyset$, and if the $A$ vertices of $Z_u$ are in $D$, then the vertex $z_i$ cannot be in $D$ since it is dominated and all its neighbors are either in $D$ or are dominated. 

So the total number of vertices in $D$ is upper-bounded by $2b + 2B + A(k - b - B) = Ak + b(2-A) + B(2-A)$. Indeed, we have the following: in a bad group, exactly two vertices are in $D$ and they can have their private neighbors in the group ; for every vertex in $D$ in a very bad group, it has one private neighbor outside the group in a good group and a single vertex is taken in $D$ in the corresponding good group ; it remains $k - b - B$ good groups in which at most $A$ vertices are in $D$, since there is $b$ bad groups, and since the private neighbor of each vertex in a very bad group is in a distinct good group where exactly one vertex is in $D$. 

But, since $A = 5 > 2$, the set $D$ has at least $k'= Ak$ vertices if and only if $b = B = 0$. So there exists no bad or very bad group in $G'$ associated to $D$. 

Since there is at most $A$ vertices in $D$ for a single good group, since there is $k$ such groups, and since $|D| \geq Ak$, it follows that $|W_i \cap D| = A$ for all $i \in \{ 1, \ldots, k \}$. 

We now construct a solution $I$ of the instance $G$ in a natural way: for any $i \in \{ 1, \ldots, k \}$, there exists a unique $u \in V_i$ such that $W_i \cap D = Z_u$, so take $u$ in the solution $I$. 

For any $u \in I$, since $Z_u \subseteq D$ and since all vertices of $Z_u$ have the same neighborhood, it follows that the vertices of $Z_u$ are their own private vertices. So, for any two $u, u' \in I$, there is no edges between the two independent sets $Z_u$ and $Z_{u'}$. So, by the construction, the set $I$ is an independent set in $G$, of size at least $k$. \qed 
\end{proof}

Now that we have proved the correctness of our reduction and since the blow-up of the reduction is linear in both the size of the instance and the size of the solution, we can now present one of the main results of this section:

\begin{theorem}\label{th:w1upperdom}
Under ETH, \textsc{$k$-Upper Dominating Set} cannot be solved in time $O(n^{o(k)})$.
\end{theorem}

\begin{proof}
Consider an instance $(G,k)$ of \textsc{$k$-Independent Set}. Apply our reduction to obtain an instance $(G',k')$ of \textsc{$k'$-Upper Dominating Set}. Thanks to Lemmas \ref{lem:forwardredisuds} and \ref{lem:backwardredisuds}, we know that $G$ has an independent set of size at least $k$ if and only if $G'$ has an upper dominating set of size at least $k' = Ak$.

Now suppose that there exists an algorithm that solves \textsc{$k'$-Upper Dominating Set} in time $O(n^{o(k')})$. With this algorithm and our reduction, we can solve \textsc{$k$-Independent Set} in time $O(n^{o(k')})$, where $k' = Ak = O(k)$, so the total running time of this procedure is $O(n^{o(k)})$, contradicting Lemma \ref{lem:w1hardis} and the ETH. \qed
\end{proof}

From now one and to obtain the FPT-approximation hardness result, we now consider our reduction above with $A$ being sufficiently large. Note that all the properties we have found before still hold since $A$ remains a constant.

Let $0 < r < 1$. To obtain the FPT-approximation hardness result for the \textsc{$k$-Independent Set} problem (see Lemma \ref{lem:fptapproxhardis}), Bonnet et al. \cite{BonnetE0P13} made a gap-amplification reduction from an instance $\phi$ of \textsc{$3$-SAT} to an instance $(G,k)$ of \textsc{$k$-Independent Set} problem. Essentially, this reduction gives the following gap:

\begin{itemize}
    \item YES-instance: If $\phi$ is satifiable, then $\alpha(G) = k$. 
    \item NO-instance: If $\phi$ is not satisfiable, then $\alpha(G) \leq rk$.
\end{itemize}

In this gap, $\alpha(G)$ is the size of a maximum independent set in $G$, and $k$ corresponds in fact to a value which depends on the reduction, but designating it by $k$ ease our purpose. 

To obtain a similar result for the \textsc{$k$-Upper Dominating Set} problem, and by using our reduction above, we have to prove that our reduction keep a gap of value $r$. Thus, we need to prove the following:

\begin{itemize}
    \item YES-instance: If $\phi$ is satisfiable, then $\alpha(G) = k$ and $\Gamma(G') = Ak$. 
    \item NO-instance: If $\phi$ is satisfiable, then $\alpha(G) \leq rk$ and $\Gamma(G') \leq rAk$. 
\end{itemize}

where $\Gamma(G')$ is the size of a maximum upper dominating set in $G'$. 

Note that we have proved the first condition in Lemma \ref{lem:forwardredisuds}, since an independent set of size at least $k$ in $G$ necessarily has size exactly $k$.

Thus, we just need to prove the second condition. To prove it, we will in fact prove the contraposition, to ease our proof. This is given in the following Lemma. The proof of this Lemma uses some arguments made in the proof of Lemma \ref{lem:backwardredisuds}, and by choosing carefully which vertices we can put in the independent set we want to construct.

\begin{lemma}\label{lem:gapisuds}
If there exists an upper dominating set in $G'$ of size $> rAk$, then there exists an independent set in $G$ of size $> rk$. 
\end{lemma}

\begin{proof}
Assume $G'$ admits an upper dominating set $D$ of size $> rAk$. We will use some properties we have proven in Lemma \ref{lem:backwardredisuds}. Recall that, for any $i \in \{ 1, \ldots, k \}$, we make the difference whether the group $W_i$ is bad, very bad, or good. Now consider a bad group $W_i$ and recall that $|W_i \cap D| = 2$ and that there exists two vertices $u, u' \in V_i$ such that $Z_u \cap D = \{ v \}$ and $Z_{u'} \cap D = \{ v' \}$. For a bad group, we will distinguish between three different types of group:

\begin{itemize}
    \item For any $j \in \{ 0, 1, 2 \}$, we say that the group $W_i$ is \textit{bad of type $j$} if there exists exactly $j$ vertices between $v$ and $v'$ which have their private neighbor outside $W_i \cup \{ z_i \}$. 
\end{itemize}

Now, given our upper dominating set $D$, we construct an independent set $I$ of $G$ as follows:

\begin{itemize}
    \item For a bad group $W_i$ of type 0, recall that there exists exactly two vertices $u, u' \in V_i$ such that $|Z_u \cap D| = 1$ and $|Z_{u'} \cap D| = 1$. Put in the solution either $u$ or $u'$. 
    \item For a bad group $W_i$ of type 1, and without loss of generality, let $u \in V_i$ such that $v = Z_u \cap D$ is the unique vertex of $W_i$ which have its private neighbor outside $W_i \cup \{ z_i \}$. Let $w$ be the private neighbor of $v$, and let $u_w$ be the vertex of $V(G)$ for which $w \in Z_{u_w}$. Put the vertex $u_w$ in the solution. 
    \item For a bad group $W_i$ of type 2, let $u, u' \in V_i$ be the two vertices such that $Z_u \cap D \neq \emptyset$ and $Z_{u'} \cap D \neq \emptyset$. Let $v = Z_u \cap D$ and $v' = Z_{u'} \cap D$, let $w$ and $w'$ be the private neighbors of $v$ and $v'$, respectively, and let $u_w$ and $u_{w'}$ be the vertices of $V(G)$ for which $w \in Z_{u_w}$ and $w' \in Z_{u_{w'}}$, respectively. Put the vertices $u_w$ and $u_{w'}$ in the solution. 
    \item For a very bad group $W_i$, let $v$ be any vertex in $W_i \cap D$. Let $w$ be the private neighbor of $v$ and let $u$ be the vertex of $V(G)$ for which $w \in Z_u$. Put $u$ in the solution. Do this for all the vertices $v \in W_i \cap D$. 
    \item For a good group $W_i$, let $u \in V_i$ such that $Z_u \cap D \neq \emptyset$. If $|Z_u \cap D| \geq 2$, then put $u$ in the solution. 
\end{itemize}

We will prove first that the solution $I$ we have constructed is an independent set of $V(G)$. 

Consider a bad group $W_i$ of type 0, and let $v = Z_u \cap D$ and $v' = Z_{u'} \cap D$. Since $W_i$ is of type 0, it means that the private neighbor of $v$ is in $Z_{u'}$ and the private neighbor of $v'$ is in $Z_u$, since $(W_i \cup \{ z_i \}) \setminus (Z_u \cup Z_{u'})$ are dominated by both $v$ and $v'$. It means that the vertices of $Z_u$ are only dominated by $v'$ and the vertices of $Z_{u'}$ are only dominated by $v$. So, since we put either $u$ or $u'$ in the solution, the vertex selected in $I$ has no neighbor in $I$. 

Consider now a bad group $W_i$ of type 1. Since $w$ is the private neighbor of $v$, it means that the vertices of $Z_{u_w}$ are only dominated by $v$. So the selected vertex $u_w$ has no neighbor in $I$. 

Consider now a bad group $W_i$ of type 2. By a similar argument as for a bad group of type 1, the vertices of $Z_{u_w}$ are only dominated by $v$, and the vertices of $Z_{u_{w'}}$ are only dominated by $v$'. So, the two selected vertices $u_w$ and $u_{w'}$ have no neighbor in $I$.

Consider now a very bad group $W_i$. We have proven in Lemma \ref{lem:backwardredisuds} that, for any vertex $v \in W_i \cap D$, it has a private neighbor outside $W_i \cup \{ z_i \}$ in an independent set $Z_u$. So the vertices of $Z_u$ are only dominated by the vertex $v$. So the selected vertex $u$ has no neighbor in $I$. This is true for all vertices $v \in W_i \cap D$. 

Consider now a good group $W_i$. Since $W_i$ is a good group, there exists at most one $u \in V_i$ such that $Z_u \cap D \neq \emptyset$. If $|Z_u \cap D| \geq 2$, then the vertices of $Z_u$ in $D$ are their own private vertices. So the selected vertex $u$ has no neighbor in $I$. 

So the solution $I$ we have constructed is an independent set in $G$. 

We will now show that $|I| > rk$.

First, consider a vertex $u \in V_i$ such that $W_i$ is bad of type 1 or 2, $Z_u \cap D = \{ v \}$, and the private neighbor of $v$ is outside $W_i \cup \{ z_i \} $. By the same arguments as for the very bad groups in Lemma \ref{lem:backwardredisuds}, we have that the private neighbor $w$ of $v$ is in a good group $W_{i'}$ and in a set $Z_{u'}$ such that $|(W_{i'} \cup \{ z_{i'} \}) \cap D| = 1$ and this unique vertex must be in $Z_{u'}$. 

Now, we give the following notations: let $b_j$ be the number of bad groups of type $j$, for $j \in \{ 0, 1, 2 \}$ ; let $B$ be the number of vertices in all the very bad groups ; let $B'$ be the number of very bad groups ; let $F$ be the number of good groups which have at least two vertices in $D$ ; and let $f$ be the number of good groups which have at most one vertex in $D$ and which do not contain private neighbors of vertices in the bad and very bad groups. 

We have $|I| \geq (b_0 + b_1 + 2b_2) + B + F$. This is easy to see from how we constructed $I$. 

Recall that there exists exactly $k$ groups $W_i$. We have the following: $F \geq k - ((b_0 + 2b_1 + 3b_2) + (B + B') + f)$. Indeed, from $k$ groups $W_i$, we substract the following: the $b_0$ bad groups of type 0 ; the $b_1$ bad groups of type 1 and the corresponding good groups (there is one such good group for each bad group of type 1) ; the $b_2$ bad groups of type 2 and the corresponding good groups (there is two such good groups for each bad group of type 2) ; the $B$ good groups which contain the private neighbors of the vertices of the very bad groups ; the $B'$ very bad groups ; and the $f$ remaining good groups. 

From this two inequalities, we obtain the following: $|I| \geq k - ((b_1 + b_2) + B' + f)$. 

Now, we will upper-bound the size of $D$ to upper-bound $B'$. Recall that for a good group counted in $F$, there is at most $A$ vertices which can be in $D$. We have: $|D| \leq (2b_0 + 3b_1 + 4b_2) + 2B + f + AF$

To see this, make the following observations: two vertices are taken for each bad group of type 0 ; two vertices and a single vertex from a good group are taken for each bad group of type 1 ; two vertices and a single vertex from two good groups are taken for each bag group of type 2 ; $B$ vertices and a single vertex from $B$ good groups are taken for the very bad groups ; at most $A$ vertices are taken for the good groups with $|W_i \cap D| \geq 2$ ; and at most one vertex is taken for each remaining good group.

But $D$ is of size $> rAk$, so we obtain:
\[ Ak + b_0(2-A) + b_1(3-2A) + b_2(4-3A) + B(2-A) - AB' + f(1-A) > rAk \]
\[ \iff B' < b_0(2/A-1) + b_1(3/A-2) + b_2(4/A-3) + B(2/A-1) + f(1/A-1) - (r-1)k \]

For $A$ sufficiently large, we obtain: $B' < - b_1 - b_2 - f - (r-1)k$

With this inequality and the one on the size of $I$, we obtain:
\[ |I| > k - ((b_1 + b_2) + f) + (b_1 + b_2 + f + (r-1)k) = k + (r-1)k = rk \]

So we have construct an independent set $I$ of $G$ of size $> rk$. \qed
\end{proof}

Now that we have proved the correctness of the gap-amplification of our reduction, we can present the second main result of this section:

\begin{theorem}\label{th:fptapproxupperdom}
Under ETH, for any constant $r > 0$ and any $\varepsilon > 0$, there is no $r$-approximation algorithm for \textsc{$k$-Upper Dominating Set} running in time $O(n^{k^{1-\varepsilon}})$.
\end{theorem}

\begin{proof}
Fix $0 < r < 1$ and $\varepsilon > 0$. Consider an instance $\phi$ of \textsc{$3$-SAT}. Apply the reduction of Bonnet et al. \cite{BonnetE0P13} to obtain an instance $(G,k)$ of \textsc{$k$-Independent Set}, and then apply our reduction to obtain an instance $(G',k')$ of \textsc{$k'$-Upper Dominating Set}. Thanks to Lemmas \ref{lem:forwardredisuds} and \ref{lem:gapisuds}, and to the gap-amplification reduction of Bonnet et al. \cite{BonnetE0P13} (see Lemma \ref{lem:fptapproxhardis}), we know the following:

\begin{itemize}
    \item YES-instance: If $\phi$ is satisfiable, then $\alpha(G) = k$ and then $\Gamma(G') = Ak$. 
    \item NO-instance: If $\phi$ is not satisfiable, then $\alpha(G) \leq rk$ and then $\Gamma(G') \leq rAk$.
\end{itemize}

Now suppose that there exists an algorithm that outputs an $r$-approximation for \textsc{$k'$-Upper Dominating Set} in time $O(n^{k^{1-\varepsilon}})$. With this algorithm and our reduction, we can obtain an $r$-approximation for \textsc{$k$-Independent Set} and thus determine if $\phi$ is satisfiable or not in time $O(n^{k^{1-\varepsilon}})$. But by Lemma \ref{lem:fptapproxhardis}, this would contradict ETH. \qed
\end{proof}

\section{Pathwidth}\label{section:pathwdith}

\subsection{FPT Algorithm Parameterized by Pathwidth}

In this section, we present an algorithm for the \textsc{Upper Dominating Set} problem parameterized by the pathwidth $pw$ of the given graph. We prove that, given a graph $G=(V,E)$ and a path decomposition $(T, \{ X_t \}_{t \in V(T)})$ of width $pw$, there exists a dynamic programming algorithm that solves \textsc{Upper Dominating Set} in time $O(6^{pw} \cdot pw)$. 

Note that Bazgan et al. have designed an FPT algorithm for \textsc{Upper Dominating Set} running in time $O^*(7^{pw})$ \cite{BazganBCFJKLLMP18}. Our algorithm essentially works as their algorithm: we have the same set of colors to give to the vertices ; and our Initialization and Forget nodes are similar to theirs. 

Nonetheless, we have modified the Introduce nodes in order to lower the complexity to $O(6^{pw} \cdot pw)$. For an Introduce node $X_t = X_{t'} \cup \{ v \}$ (and a vertex $v \notin X_{t'}$), Bazgan et al. did the following: they go through all possible colorings of the bag $X_t$ and consider every subset of the neighborhood of $v$ to give the right color to the vertices of this subset. Thus, since they consider every subset of the neighborhood of $v$, they get an algorithm running in time $O^*(7^{pw})$. 

In our algorithm, we do the following: for an Introduce node $X_t = X_{t'} \cup \{ v \}$, we go through all possible colorings of the bag $X_{t'}$ and through all colorings of the vertex $v$, and we update the value in the table depending on the corresponding colorings of $X_{t'}$ and $v$. Doing so, and by being careful on the color given to $v$, it enables us to get an algorithm running in time $O(6^{pw} \cdot pw)$. We obtain the following Theorem:

\begin{theorem}\label{th:positiveupperdom}
The \textsc{Upper Dominating Set} problem can be solved in time $O(6^{pw} \cdot pw)$, where $pw$ is the input graph's pathwidth. 
\end{theorem}

\begin{proof}
We now suppose that we are given a path decomposition $(T, \{ X_t \}_{t \in V(T)})$ of the given graph $G = (V,E)$. Recall that in such a path decomposition, we have three types of bag: the \textit{Initialization} bag, the \textit{Forget} bags, and the \textit{Introduce} bags. We can assume that we are given a nice path decomposition. So we only need to describe the Initialization, Forget and Introduce nodes, where a vertex is introduced exactly once, and is forgotten exactly once. 


We will now present how our dynamic programming works for each type of bag. To do so, we first distinguish between six different colors for each vertex. We define a \textit{coloring} of a bag $X_t$ to be a mapping $f : X_t \to \{ I, F, F^*, O^*, O, P \}$ assigning six different colors to the vertices of the bag $X_t$. These six colors are defined as follows:
\begin{itemize}
    \item $I$: the set of vertices which are in the dominating set and which forms an independent set, i.e. the vertices of the solution which are their own private vertices. 
    \item $F$: the set of vertices which are in the dominating set and which are already matched to a private neighbor. 
    \item $F^*$: the set of vertices which are in the dominating set and which have no private neighbor yet. 
    \item $O^*$: the set of vertices which are not in the dominating set and which are not dominated yet. 
    \item $O$: the set of vertices which are not in the dominating set, which are dominated, but which are not private neighbors of vertices of the solution. 
    \item $P$: the set of vertices which are not in the dominating set, which are dominated, and which are private neighbors of some vertices of the solution. 
\end{itemize}

Note that, since $f^{-1}(I) \cup f^{-1}(F) \cup f^{-1}(F^*) \cup f^{-1}(O^*) \cup f^{-1}(O) \cup f^{-1}(P)$ is a partition of $X_t$, there are $6^{|X_t|}$ colorings of $X_t$. These colorings form the space of states of the node $X_t$, and we will use this fact to improve the algorithm of Bazgan et al. \cite{BazganBCFJKLLMP18} from $O^*(7^{pw})$ to $O(6^{pw} \cdot pw)$. 

For a coloring $X_t$, we denote by $c[t,f]$ the maximum size of an upper dominating set $D \subseteq V_t$ (where $V_t$ denotes the set of vertices belonging to any bag $X_t$ of the subtree $T_t$ rooted at the node $t$), such that:

\begin{itemize}
    \item $f^{-1}(I) \cup f^{-1}(F) \cup f^{-1}(F^*) = D$
    \item $f^{-1}(O) \cup f^{-1}(P)$ is dominated by $D$. 
\end{itemize}

We call such a set $D$ a \textit{maximum compatible} set for $t$ and $f$. If no maximum compatible set for $t$ and $f$ exists, then we put $c[t,f] = - \infty$. 

Let us now define some useful notations. For a subset $X \subseteq V$, consider a coloring $f : X \to \{ I, F, F^*, O^*, O, P \}$. For a vertex $v \in V$, and a color $\alpha \in \{ I, F, F^*, O^*, O, P \}$, we define a new coloring $f_{v \to \alpha} : X \cup \{ v \} \to \{ I, F, F^*, O^*, O, P \}$ as follows:
\[ f_{v \to \alpha}(x) =
\begin{cases}
f(x) & \text{if } x \neq v \\
\alpha & \text{if } x = v \\
\end{cases}
\]

We now proceed to present the recursive formulas for the values of $c$.

\textbf{Initialization node}. For a node $X_t = \{ v \}$ which initializes the table, we make the following observations: the vertex $v$ cannot be in $F$ since it cannot have a private neighbor ; it cannot be neither in $O$ nor in $P$ since it cannot be dominated ; and for the three other cases (for $I$, $F^*$ and $O^*$), we just have to give the size of the corresponding solution. We obtain:
\[ c[t,f] =
\begin{cases}
1 & \text{if } v \in I \cup F^* \\
0 & \text{if } v \in O^* \\
-\infty & \text{otherwise} \\
\end{cases}
\]

\textbf{Forget node}. Let $t$ be a forget node with a unique child $t'$ such that $X_t = X_{t'} \setminus \{ v \}$ for some $v \in X_{t'}$. We make the following observations: the vertex $v$ cannot be forgotten if it belongs to $F^*$ since it contradicts the fact that the solution is minimal ; $v$ cannot be forgotten if it belongs to $O^*$ since in this case it remains undominated ; the four other cases are valid and we just need to take the maximum value between these four cases. We obtain:
\[ c[t,f] = \max \{ c[t',f_{v \to I}], c[t',f_{v \to F}], c[t',f_{v \to O}], c[t',f_{v \to P}] \} \]

Note that for these Initialization and Forget nodes, since in the worst case we go through all possible colorings of the bag $X_t$, the running-time for these two types of bags is $O(6^{pw})$.

\textbf{Introduce node}.  Let $t$ be an introduce node with a unique child $t'$ such that $X_t = X_{t'} \cup \{ v \}$ for some $v \notin X_{t'}$. Here, instead of going through all possible coloring of the bag $t$ and considering every subset of the neighborhood of $v$ to put in $O^*$, as Bazgan et al. did, we go through all possible colorings of the bag $t'$ and update the value of $c[t,f]$ depending on the corresponding coloring and any color affected to $v$. This enables us to lower the complexity to $O(6^{pw} \cdot pw)$ since we don't need anymore to go through every subset of the neighborhood of $v$. First, we affect the value $-\infty$ to $c[t,f]$ for every coloring $f$ of the bag $t$. Then, for every coloring $f : X_{t'} \to \{ I, F, F^*, O^*, O, P \}$ and any color $\alpha : \{ v \} \to \{ I, F, F^*, O^*, O, P \}$, we will define a new coloring $f_{new}$ which corresponds to the coloring $f$ of the vertices of $X_{t'}$ plus the coloring $\alpha$ of $v$, and an associated value $k_{new}$ which will be the size of the corresponding upper dominating set. To get the final value $c[t,f_{new}]$ for the bag $t$, we just update it by $k_{new}$ if $c[t,f_{new}] < k_{new}$ for the new coloring $f_{new}$, so at the end each entry of the table of $t$ will have the maximum size for the corresponding coloring $f_{new}$. Now, for each coloring $f$ of the bag $t'$ and each color $\alpha \in \{ I, F, F^*, O^*, O, P \}$ of the vertex $v$, we have the following cases: 
\begin{itemize}
    \item If $\alpha = I$ and $N(v) \cap (f^{-1}(I) \cap f^{-1}(F) \cap f^{-1}(F^*) \cap f^{-1}(P)) = \emptyset$, then:
\end{itemize}
\[ f_{new} (u) = 
\begin{cases}
O & \text{if } f(u) = O^* \text{ and } (u,v) \in E \\
f(u) & \text{otherwise } \\
\end{cases}
\]

For $v$ to be in $I$, we need that all its neighbors are either in $O$ or in $O^*$. Then, if this condition is satisfied, we can give the color $O$ to all neighbors of $v$ which are not dominated in the bag $t'$ since they become dominated in the new bag. 
\begin{itemize}
    \item If $\alpha = F$ and $N(v) \cap (f^{-1}(I) \cup f^{-1}(P)) = \emptyset$ and $\exists w \in N(v) \cap f^{-1}(O^*)$, then:
\end{itemize}
\[ f_{new} (u) = 
\begin{cases}
P & \text{if } u = w \\
O & \text{if } f(u) = O^* \text{ and } u \neq w \text{ and } (u,v) \in E \\
f(u) & \text{otherwise } \\
\end{cases}
\]

For $v$ to be in $F$, note that its neighbors cannot be in $I$ or in $P$, because otherwise its neighbors in $I$ are dominated and its neighbors in $P$ cannot be private neighbors anymore. Note also that at least one neighbor of $v$ has to be in $O^*$ in the bag $t'$ in order to become the private neighbor of $v$. Moreover, $v$ belongs to $F$ if at least one of its neighbors belongs to $P$ in the new coloring, so we can take any neighbor of $v$ which is not dominated in $t'$ and put it in $P$ in the new coloring, since it is enough to have just one neighbor of $v$ being its private neighbor. For the other neighbors of $v$ which are not dominated in $t'$, we can give them the color $O$ since they become dominated. Finally, all other vertices keep the same color from the coloring $f$. 
\begin{itemize}
    \item If $\alpha = F^*$ and $N(v) \cap (f^{-1}(I) \cup f^{-1}(P)) = \emptyset$, then:
\end{itemize}
\[ f_{new} (u) = 
\begin{cases}
O & \text{if } f(u) = O^* \text{ and } (u,v) \in E \\
f(u) & \text{otherwise } \\
\end{cases}
\]

For $v$ to be in $F^*$, we need that its neighbors are neither in $I$ nor in $P$, because otherwise its neighbors in $I$ are dominated and its neighbors in $P$ cannot be  private neighbors anymore. If this condition is satisfied, we can give the color $O$ to all neighbors of $v$ which are not dominated in the bag $t'$ since they become dominated in the new bag. 
\begin{itemize}
    \item If $\alpha = O^*$ and $N(v) \cap (f^{-1}(I) \cup f^{-1}(F) \cup f^{-1}(F^*)) = \emptyset$, then: $f_{new}(u) = f(u) \forall u \in X_{t'}$
\end{itemize}

For $v$ to be in $O^*$, we just need to check that all its neighbors are not in the solution, that is they do not have the colors $I$, $F$ or $F^*$. If this condition is satisfied, $v$ can be added as a non dominated vertex and the coloring $f_{new}$ is the coloring $f$. 
\begin{itemize}
    \item If $\alpha = O$ and $N(v) \cap (f^{-1}(I) \cup f^{-1}(F) \cup f^{-1}(F^*)) \neq \emptyset$, then: $f_{new}(u) = f(u) \forall u \in X_{t'}$
\end{itemize}

For $v$ to be in $O$, we just need that at least one of its neighbors is in the solution, that is one of its neighbor is in $I$, in $F$ or in $F^*$. If this condition is satisfied, $v$ can get color $O$ and the coloring $f_{new}$ is the coloring $f$. 
\begin{itemize}
    \item If $\alpha = P$ and $N(v) \cap (f^{-1}(I) \cup f^{-1}(F)) = \emptyset$ and $\exists ! w \in N(v) \cap f^{-1}(F^*)$, then:
\end{itemize}
\[ f_{new} (u) = 
\begin{cases}
F & \text{if } u = w \\
f(u) & \text{otherwise } \\
\end{cases}
\]

For $v$ to have color $P$, we need firstly that its neighbors are neither in $I$ nor in $F$, because otherwise $v$ cannot be a private neighbor of some vertex of the solution, and we also need that exactly one neighbor of $v$ is in $F^*$ in the bag $t'$, so that these two vertices are matched. If these conditions are satisfied, we give color $F$ to the only neighbor of $v$ in $F^*$ in the bag $t'$ and all other vertices keep the same color from the coloring $f$. 


Note that, since we go through all possible colorings $f$ in the bag $t'$ and through all colors for the vertex $v$, the running time of any introduce bag is $O(6^{pw} \cdot pw)$. \qed
\end{proof}

\subsection{Lower Bound}

In this section, we present a lower bound on the complexity of any FPT algorithm for the \textsc{Upper Dominating Set} problem parameterized by the pathwidth of the graph matching our previous algorithm. More precisely, we prove that, under SETH, for any $\varepsilon > 0$, there is no algorithm for \textsc{Upper Dominating Set} running in time $O^*((6-\varepsilon)^{pw})$, where $pw$ is the pathwidth of the input graph. 

To get this result, we will do a reduction from the \textsc{$q$-CSP-6} problem (see \cite{Lampis18}) to the \textsc{Upper Dominating Set} problem. In the former problem, we are given a \textsc{Constraint Satisfaction} (CSP) instance with $n$ variables and $m$ constraints. The variables take values over a set of size 6. Without loss of generality, let $\{ 0, 1, 2, 3, 4, 5 \}$ be this set. Each constraint involves at most $q$ variables, and is given as a list of acceptable assignments for these variables, where an acceptable assignment is a $q$-tuple of values from the set $\{ 0, 1, 2, 3, 4, 5 \}$ given to the $q$ variables. Without loss of generality, we force the following condition: each constraint involves exactly $q$ variables, because if it has fewer, we can add to it new variables and augment the list of satisfying assignments so that the value of the new variables is irrelevant. 


The following result, shown in \cite{Lampis18}, is a natural consequence of the SETH, and will be the starting point to obtain the desired lower bound:

\begin{lemma}[Lemma 2 from \cite{Lampis18}]\label{lem:lampis}
If the SETH is true, then, for all $\varepsilon > 0$, there exists a $q$ such that $n$-variables \textsc{$q$-CSP-6} cannot be solved in time $O^*((6-\varepsilon)^{n})$. 
\end{lemma}

We note that in \cite{Lampis18}, it was shown that for any constant $B$, \textsc{$q$-CSP-B} cannot be solved in time $O^*((B-\varepsilon)^n)$ under the SETH. For our purpose, only the case where $B = 6$ is relevant because this corresponds to the base of our target lower bound. 

We will produce a polynomial time reduction from an instance of \textsc{$q$-CSP-6} with $n$ variables to an equivalent instance of \textsc{Upper Dominating Set} whose pathwidth is bounded by $n + O(1)$. Thus, any algorithm for the latter problem running faster than $O^*((6-\varepsilon)^{pw})$ would give a $O^*((6-\varepsilon)^n)$ algorithm for the former problem, contradicting SETH. 

Before we proceed further in the description of our reduction, let us give the basic ideas, which look like other SETH-based lower bounds from the literature \cite{HanakaKLOS18,JaffkeJ17,KatsikarelisLP18,KatsikarelisLP17,LokshtanovMS18}. The constructed graph consists of a main part of $n$ paths of length $4m$, each divided into $m$ sections. The idea is that an optimal solution will verify, for each path, a specific pattern in the whole graph. For four consecutive vertices, there are six ways for taking exactly two vertices among the four and dominating the two others. These six ways for each path will represent all possible assignments for all variables. Then, we will add some \textit{verification} gadgets for each constraint and attach it to the corresponding section, in order to check that the selected assignment satisfies the constraint or not. 

A first difficulty of this reduction is to prove that an optimal solution of the \textsc{Upper Dominating Set} instance has the desired form, and more precisely that the pattern selected for a variable is constant throughout the graph. To answer this difficulty, and by using a technique introduced in \cite{LokshtanovMS18}, we make a polynomial number of copies of this construction and we connect them together, enabling us to have a sufficiently large copy where the patterns are kept constant in this copy.

Moreover, we need to be careful in our verification gadgets in order to have the following conditions: the vertices of the paths taken in the solution must not have any private neighbor in the corresponding verification gadget, because otherwise it would be impossible to keep the patterns constant in a sufficiently large copy of the graph; and the vertices of the paths not taken in the solution must not be dominated by the corresponding verification gadget, because otherwise there can be some vertices of the paths taken in the solution that have no private neighbor. 

\subsubsection*{Construction}

Let us now present our reduction. We are given a \textsc{$q$-CSP-6} instance $\varphi$ with $n$ variables $x_1, \ldots, x_n$ taking values over the set $\{ 0, 1, 2, 3, 4, 5 \}$, and $m$ constraints $c_0, \ldots, c_{m-1}$, each containing exactly $q$ variables and $C_j$ possible assignments over these $q$ variables, for each $j \in \{ 0, \ldots, m-1 \}$. We define the following numbers: $A = 4q+2$ and $F = (2n+1)(4n+1)$. We set our budget to be $k = Fm(2n + A) + 2n$. 

We construct our instance of \textsc{Upper Dominating Set} as follows:
\begin{enumerate}
    \item For $i \in \{ 1, \ldots, n \}$, we construct a path $P_i$ of $4Fm+6$ vertices: the vertices are labeled $u_{i,j}$ for $j \in \{ -3, \ldots, 4Fm+2 \}$ ; and for each $i,j$ the vertex $u_{i,j}$ is connected to $u_{i,j+1}$. We call these paths the \textit{main} part of our graph. 
    \item For each section $j \in \{ 0, \ldots, Fm-1 \}$, let $j' = j \mod m$. We construct a verification gadget $H_j$ as follows:
    \begin{enumerate}
        \item A clique $K_{j}$ of size $AC_{j'}$ such that the $AC_{j'}$ vertices are partitioned into $C_{j'}$ cliques $K_{j}^1, \ldots, K_{j}^{C_{j'}}$, each corresponding to a satisfying assignment $\sigma_l$ in the list of $c_{j'}$, for $l \in \{ 1, \ldots, C_{j'} \}$, and each containing exactly $A$ vertices. 
        \item A clique $L_j$ of size $AC_{j'}$ such that the $AC_{j'}$ vertices are partitioned in $C_{j'}$ cliques $L_{j}^1, \ldots, L_{j}^{C_{j'}}$, each containing exactly $A$ vertices. 
        \item For each $i \in \{ 1, \ldots, n \}$ such that $x_i$ is involved in $c_{j'}$, and for each satisfying assignment $\sigma_l$ in the list of $c_{j'}$: if $\sigma_l$ sets $x_i$ value 0, connect the two vertices $u_{i,4j+2}$ and $u_{i,4j+3}$ to the $A$ vertices of the clique $K_{j}^l$ ;  if $\sigma_l$ sets $x_i$ value 1, connect the two vertices $u_{i,4j+3}$ and $u_{i,4j}$ to the $A$ vertices of the clique $K_{j}^l$ ; if $\sigma_l$ sets $x_i$ value 2, connect the two vertices $u_{i,4j}$ and $u_{i,4j+1}$ to the $A$ vertices of the clique $K_{j}^l$ ; if $\sigma_l$ sets $x_i$ value 3, connect the two vertices $u_{i,4j+1}$ and $u_{i,4j+2}$ to the $A$ vertices of the clique $K_{j}^l$ ; if $\sigma_l$ sets $x_i$ value 4, connect the two vertices $u_{i,4j+1}$ and $u_{i,4j+3}$ to the $A$ vertices of the clique $K_{j}^l$ ; if $\sigma_l$ sets $x_i$ value 5, connect the two vertices $u_{i,4j}$ and $u_{i,4j+2}$ to the $A$ vertices of the clique $K_{j}^l$. 
        \item For each satisfying assignment $\sigma_l$ in the list of $c_{j'}$, do the following: add a matching between the vertices of $K_{j}^l$ and the vertices of $L_{j}^l$ ; for any $l' \in \{ 1, \ldots, C_{j'} \}$ with $l' \neq l$, add all the edges between the vertices of $K_{j}^l$ and the vertices of $L_{j}^{l'}$. 
        \item Add a vertex $w$ connected to all the vertices of the clique $L_j$.
    \end{enumerate}
\end{enumerate}

Now that we have presented our reduction, we argue that it is correct and that the obtained graph $G$ has the desired pathwidth. Recall that the target size of an optimal solution in $G$ is $k$ as defined above. 

\begin{lemma}\label{lem:forwardupperdom}
If $\varphi$ is satisfiable, then there exists an upper dominating set in $G$ of size at least $k$. 
\end{lemma}

\begin{proof}
Assume $\varphi$ admits some satisfying assignment $\rho : \{ x_1, \ldots, x_n \} \to \{ 0, 1, 2$, $3, 4, 5 \}$. We construct a solution $S$ of the instance $G$ of \textsc{Upper Dominating Set} as follows:

\begin{enumerate}
    \item For each $i \in \{ 1, \ldots, n \}$, let $\alpha$ and $\beta$ be the following numbers: if $\rho(x_i) = 0$, let $\alpha = 2$ and $\beta = 3$ ; if $\rho(x_i) = 1$, let $\alpha = 3$ and $\beta = 0$ ; if $\rho(x_i) = 2$, let $\alpha = 0$ and $\beta = 1$ ; if $\rho(x_i) = 3$, let $\alpha = 1$ and $\beta = 2$ ; if $\rho(x_i) = 4$, let $\alpha = 1$ and $\beta = 3$ ; if $\rho(x_i) = 5$, let $\alpha = 0$ and $\beta = 2$. Let $U = \bigcup_{j = 0}^{Fm-1 } \{ u_{i,4j+\alpha}, u_{i,4j+\beta} \}$. We add to the solution all vertices of $(V(P_i) \setminus \{ u_{i,-3}, u_{i,-2}, u_{i,-1}, u_{i,4Fm}, u_{i,4Fm+1}, u_{i,4Fm+2} \}) \setminus U$. 
    \item For each $j \in \{ 0, \ldots, Fm-1 \}$, let $j' = j \mod m$. Consider the unique possible assignment $\sigma_{l^*}$ in the list of $c_{j'}$ satisfied by $\rho$ (such a unique possible assignment must exist since $\rho$ satisfies $\varphi$), and take the $A$ vertices of the clique $L_{j}^{l^*}$.
    \item For each $i \in \{ 1, \ldots, n \}$, do the following: if $\rho(x_i) = 0$, then add $u_{i,-3}$, $u_{i,4Fm}$ and $u_{i,4Fm+1}$ to $S$ ; if $\rho(x_i) = 1$, then add $u_{i,-2}$ and $u_{i,4Fm+1}$ to $S$ ; if $\rho(x_i) = 2$, then add $u_{i,-2}$, $u_{i,-1}$ and $u_{i,4Fm+2}$ to $S$ ; if $\rho(x_i) = 3$, then add $u_{i,-3}$ and $u_{i,4Fm+2}$ to $S$ ; if $\rho(x_i) = 4$, then add $u_{i,-3}$ and $u_{i,4Fm+1}$ to $S$ ; if $\rho(x_i) = 5$, then add $u_{i,-2}$ and $u_{i,4Fm+2}$ to $S$. 
\end{enumerate}

Let us now argue why this solution has size at least $k$. In the first step, we have selected $2Fmn$ vertices. To see this, let $Q_{i,j}$ be the sub-path of $P_i$ corresponding to the section $j$ ($j \in \{ 0, \ldots, Fm-1 \}$), i.e. $Q_{i,j} = \{ u_{i,4j}, u_{i,4j+1}, u_{i,4j+2}, u_{i,4j+3} \}$. Observe that we have put exactly two vertices of $Q_{i,j}$ in $U$, which leaves two vertices in the solution, for all $i$ and all $j$. Consider now any $j \in \{ 0, \ldots, Fm-1 \}$ and the corresponding verification gadget $H_j$. In this gadget, we have selected all the vertices of the clique $L_{j}^{l^*}$, corresponding to the satisfied assignment $\sigma_{l^*}$. So we have selected $AFm$ vertices for all the verification gadgets. Finally, at least $2n$ vertices have been added to the solution at step 3. So the total size is at least $2Fmn + AFm + 2n = k$. 

Let us now argue why the solution is a valid upper dominating set. 

Consider any $j \in \{ 0, \ldots, Fm-1 \}$ and let $j' = j \mod m$. We have selected the $A$ vertices of the clique $L_{j}^{l^*}$ corresponding to the unique possible assignment $\sigma_{l^*}$ in the list of $c_{j'}$ satisfied by $\rho$ (such a unique possible assignment must exist since $\rho$ satisfies $\varphi$). Since $L_{j}$ is a clique, since the vertices of $L_{j}^{l^*}$ are connected to all vertices of $K_{j}
^{l'}$, for any $l' \in \{ 1, \ldots, C_{j'} \}$ with $l' \neq l^*$, since there is a matching between the vertices of $L_{j}^{l^*}$ and the vertices of $K_{j}^{l^*}$, and since the vertex $w$ is connected to all vertices of $L_j$, we have that all the vertices of $H_j$ are dominated by $S$. 

Now, observe that, since $\sigma_{l^*}$ is satisfied by $\rho$, it means that the values given by $\rho$ to the variables appearing in the constraint $c_{j'}$ satisfy $\sigma_{l^*}$, so by the construction it follows that the neighbors of the vertices of $K_{j}^{l^*}$ in the paths all belongs to $U$. Indeed, consider any variable $x_i$ appearing in $c_{j'}$: if $\sigma_{l^*}$ sets value 0 to $x_i$, then $\rho(x_i) = 0$, and then, for $\alpha = 2$ and $\beta = 3$, we have that $u_{i,4j+\alpha}$ and $u_{i,4j+\beta}$ are in $U$ and are the only vertices of $Q_{i,j}$ neighbors of the vertices of $K_{j}^{l^*}$ ; it remains true whether $\sigma_{l^*}$ sets value 1, 2, 3, 4 or 5 to $x_i$ with the convenient $\alpha$ and $\beta$. So all the neighbors of $K_{j}^{l^*}$ in the main part of the graph are not in $S$. Moreover, no vertex of $K_j$ is taken in the solution, and no vertex of $L_j \setminus L_{j}^{l^*}$ is taken in the solution. By these facts, and since the only edges between $L_{j}^{l^*}$ and $K_{j}^{l^*}$ is a perfect matching between the vertices of these two sets, it follows that each vertex of $L_{j}^{l^*}$ has a private neighbor, namely its unique neighbor in $K_{j}^{l^*}$. 

Consider now any $i \in \{ 1, \ldots, n \}$. The set $U$ never takes three consecutive vertices in the path $P_i$, so $(V(P_i) \setminus \{ u_{i,-3}, u_{i,-2}, u_{i,-1}, u_{i,4Fm}, u_{i,4Fm+1}, u_{i,4Fm+2} \}) \setminus U$ is a dominating set in the path $(V(P_i) \setminus \{ u_{i,-3}, u_{i,-2}, u_{i,-1}, u_{i,4Fm}, u_{i,4Fm+1}$, $u_{i,4Fm+2} \})$. Observe now that, for any $j \in \{ 0, \ldots, Fm-1 \}$, the vertices of the clique $K_j$ in the gadget $H_j$ are never taken by the solution, so the vertices of the path $P_i$ are only dominated by the vertices of $P_i$, whether the variable $x_i$ appears in $c_{j'}$ or not (for $j' = j \mod m$). Moreover, by the same argument, the neighbors in the verification gadgets of the vertices of the path $P_i$ taken in the solution are never taken in the solution. 

If $\rho(x_i) \in \{ 0, 1, 2, 3 \}$, then $U$ takes two consecutive vertices, leaves two consecutive vertices in $S$, takes again two consecutive vertices, and so on. In these cases, the two vertices of $S$ each have a private neighbor, namely their other neighbor in the path. If $\rho(x_i) \in \{ 4, 5 \}$, then $U$ takes a vertex, leaves a vertex in $S$, takes a vertex, and so on. In these cases, the vertices of $S$ are their own private vertex. So all the vertices of the path either have a private neighbor, or are their own private vertices. 

Nonetheless, we have to be more careful for the first and last sections (for $j = 0$ and $j = Fm-1$). By the step 3 of our construction of the solution $S$, and by some simple observations, we have that all vertices of the main part are dominated, and that the vertices of the main part which belong to the solution either have a private neighbor in the corresponding path, or are their own private vertices. \qed
\end{proof}

Let us now prove the other direction of our reduction. The idea of this proof is the following: by partitioning the graph into different parts and upper bound the cost of these parts, we prove that if an upper dominating set in $G$ has not the same form as in Lemma \ref{lem:forwardupperdom} in a sufficiently large copy, then it has size strictly less than $k$, enabling us to produce a satisfiable assignment for $\varphi$ using the copy where the upper dominating set has the desired form. 

\begin{lemma}\label{lem:barwardupperdom}
If there exists an upper dominating set of size at least $k$ in $G$, then $\varphi$ is satisfiable. 
\end{lemma}

\begin{proof}
Suppose that we are given an upper dominating set $D$ of maximum size. Before we proceed any further, let us define, for each $S \subseteq V$, its \textit{cost} as $cost(S) = |S \cap D|$. Clearly, $cost(V) \geq k$. Also, for two disjoint sets $S_1$ and $S_2$, we have $cost(S_1 \cup S_2) = cost(S_1) + cost(S_2)$. Our strategy will therefore be to partition $V$ into different parts and upper bound their cost. 

For each $j \in \{ 0, \ldots, Fm-1 \}$, let $V_j = H_j \cup \bigcup_{i=1}^n Q_{i,j}$, where $Q_{i,j} = \{ u_{i,4j}, u_{i,4j+1}, u_{i,4j+2}, u_{i,4j+3} \}$. 

\begin{claim}
$\cost(V_j) \leq 2n + A$.
\end{claim}

\begin{proof}

Consider any $j \in \{ 0, \ldots, Fm-1 \}$, and let $j' = j \mod m$. We will prove that $cost(H_j) \leq A$. Note that the vertex $w$ has to be dominated, so either it is in $D$, or at least one vertex of $L_j$ is in $D$. 

First, suppose that the vertex $w$ belongs to $D$. No vertex of $L_j$ can be in $D$, because otherwise $w$ has no private neighbor and is the neighbor of another vertex of $D$. Moreover, since $L_j$ is dominated, either only one vertex of $K_j$ belongs to $D$ and all the other vertices of $K_j$ can be its private neighbor, and in this case, the desired bound is obtained, or more that one vertex of $K_j$ belongs to $D$. In this case, since $K_j$ is a clique, and since $L_j$ is dominated, the vertices in $D \cap K_j$ must have their private neighbor in the main part of the graph. Note first that, for any $l \in \{ 1, \ldots, C_{j'} \}$, it cannot be the case that two vertices of $K_{j}^{l}$ are in $D$, since they share the same neighborhood. So the vertices of $K_j$ that belongs to $D$ are in at least two distinct cliques $K_j^{l_1}$ and $K_{j}^{l_2}$ (for $l_1, l_2 \in \{ 1, \ldots, C_{j'} \}$ and $l_1 \neq l_2$). Note that, for any $i \in \{ 1, \ldots, n \}$ such that $x_i$ is involved in $c_{j'}$, any vertex of $K_j$ is connected to two vertices of $Q_{i,j}$ (for $Q_{i,j} = \{ u_{i,4j}, u_{i,4j+1}, u_{i,4j+2}, u_{i,4j+3} \}$). So it cannot be the case that three vertices of $K_j$ are in $D$, because it would imply that one of them has to private neighbor. So if the vertex $w$ is in $D$, then we have $cost(H_j) \leq 3$. 

Let us now consider the case where $w$ does not belong to $D$. Suppose now that there exists $l_1, l_2 \in \{ 1, \ldots, C_{j'} \}$ with $l_1 \neq l_2$ such that at least two vertices of $L_{j}^{l_1}$, let say $v_1$ and $v_1'$, and at least one vertex of $L_{j}^{l_2}$, let say $v_2$, belong to $D$. Note that, since $L_j \cup \{ w \}$ is a clique, the three vertices $v_1$, $v_1'$ and $v_2$ must have, each of them, a private neighbor in $K_j$. Now observe that all the vertices of $L_{j}^{l_1}$ are connected to all vertices of $K_j \setminus K_j^{l_1}$, so the private neighbors of $v_1$ and $v_1'$ must belong to $K_j^{l_1}$. But the vertex $v_2$ is connected to all vertices of $K_j^{l_1}$, since all vertices of $L_j^{l_2}$ are, which implies that $v_1$ and $v_1'$ have no private neighbor. So it cannot be the case that at least two vertices of $L_j^{l_1}$ and at least one vertex of $L_j^{l_2}$ are in $D$, for any $l_1, l_2$. 

Suppose now that there exists $l_1, l_2, l_3 \in \{ 1, \ldots, C_{j'} \}$ with $l_1 \neq l_2 \neq l_3$ such that one vertex of $L_j^{l_1}$, let say $v_1$, one vertex of $L_j^{l_2}$, let say $v_2$, and one vertex of $L_j^{l_3}$, let say $v_3$, are in $D$. By a similar argument, we have that the private neighbor of $v_1$ has to be in $K_j^{l_1}$: it cannot be in $K_j^{l_2}$ since all vertices of $K_j^{l_2}$ are connected to $v_1$ and $v_3$ ; it cannot be in $K_j^{l_3}$ since all vertices of $K_j^{l_3}$ are connected to $v_1$ and $v_2$ ; and it cannot be in any other $K_j^{l'}$ (for $l' \neq l_1, l_2, l_3$) since the vertices of $K_j^{l'}$ are connected to $v_1$, $v_2$ and $v_3$. But observe that all the vertices of $K_j^{l_1}$ are connected to $v_2$ and $v_3$, which implies that $v_1$ has no private neighbor. So it cannot be the case that one vertex of $L_j^{l_1}$, one vertex of $L_j^{l_2}$ and one vertex of $L_j^{l_3}$, are in $D$, for any $l_1, l_2, l_3$. 

So, by these arguments, we have that at most $A$ vertices of $L_j \cup \{ w \}$ belong to $D$, i.e. the $A$ vertices of a single clique $L_j^{l}$ (for $l \in \{ 1, \ldots, C_{j'} \}$). Now, suppose that there exist $l \in \{ 1, \ldots, C_{j'} \}$ such that $D \cap L_j = L_j^{l}$. The private neighbors of these vertices taken in $D$ must be in $K_j^{l}$, which implies that no vertex of $K_j$ can be in $D$. It follows that $cost(H_j) \leq A$, and this bound is attained if there exists an $l \in \{ 1, \ldots, C_{j'} \}$ such that $L_j^{l} \subseteq D$ and such that the vertices of $K_j^{l}$ are only dominated by the vertices of $L_j^{l}$.  

Now, consider any $j \in \{ 0, \ldots, Fm-1 \}$ and any $i \in \{ 1, \ldots, n \}$ such that variable $x_i$ is involved in $c_{j'}$, for $j' = j \mod m$. Suppose that at least three vertices of $Q_{i,j}$ are in $D$, where we recall $Q_{i,j} = \{ u_{i,4j}, u_{i,4j+1}, u_{i,4j+2}, u_{i,4j+3} \}$. Then all vertices of $K_j$ are dominated, since every vertex of $K_j$ is connected to two vertices of $Q_{i,j}$. From this it follows that at most one vertex of $H_j$ is in $D$, since $L_j \cup \{ w \}$ is a clique. Let $W_j = H_j \cup \bigcup_{x_i \text{active}} Q_{i,j}$. We have $cost(W_j) \leq 4q + 1$. We construct another solution by doing the following: consider a satisfying assignment $\sigma_l$ in the list of $c_{j'}$ and take all vertices of the clique $L_j^{l}$ ; plus take all the vertices of $Q_{i,j}$ not neighbors of the vertices of $K_j^{l}$, for any active variable $x_i$ ; and modify the solution to obtain an upper dominating set. Clearly, it gives us a valid solution. Moreover, this has increase the total cost. Indeed, we lose at most $4q+1$ vertices: at most $2$ vertices per $Q_{i,j}$ if the original solution had taken the four vertices ; at most the two vertices $u_{i,4j-1}$ and $u_{i,4(j+1)}$, for each $x_i$ active, in order to keep the solution valid ; and the vertex of $L_j \cup \{ w \}$. On the other side, we have added $4q+2$ vertices: the $A = 4q+2$ vertices of $L_j^l$. Doing so should not be possible since $D$ is of maximum size, so for any active variable $x_i$, at most two vertices of $Q_{i,j}$ belong to $D$. 

Now, consider any $j \in \{ 0, \ldots, Fm-1 \}$ and any $i \in \{ 1, \ldots, n \}$ such that variable $x_i$ is not involved in $c_{j'}$, for $j' = j \mod m$. Observe that, since the vertices of $Q_{i,j}$ are not connected to any verification gadget, it cannot be the case that three vertices of $Q_{i,j}$ belong to $D$, because otherwise at least one of them would be neighbor of another vertex of $D$ and would have no private neighbor. 

We now have all the lower bounds we need: $cost(H_j) \leq A$ ; and $cost(Q_{i,j}) \leq 2$, whether $x_i$ is active or not. So $cost(V_j) \leq 2n + A$. \qed
\end{proof}


We will say that $j$ is \textit{problematic} if $cost(V_j) < 2n + A$. 

Now, consider any $i \in \{ 1, \ldots, n \}$ and observe that among the three vertices $u_{i,-3}, u_{i,-2}$ and $u_{i,-1}$, at most two vertices can be in $D$, because otherwise the vertex $u_{i,-3}$ has no private neighbor. The same observation holds for the last three vertices $u_{i,4Fm}, u_{i,4Fm+1}$ and $u_{i,4Fm+2}$. 

Let $L \subseteq \{ 0, \ldots, Fm-1 \}$ be the set of problematic indices. We claim that $|L| \leq 2n$. Indeed, we have $cost(V) \leq \sum_{j=0}^{Fm-1} cost(V_j) + 4n \leq Fm(2n+A) - |L| + 4n = k + 2n - |L|$. But since the total cost is at least $k$, we have $|L| \leq 2n$. Now consider the longest contiguous interval $K \subseteq \{ 0, \ldots, Fm-1 \}$ such that all $j \in K$ are not problematic. Since $F = (2n+1)(4n+1)$, we have $K \geq Fm / (|L|+1) = m(4n+1)$. 

Before we proceed further, note that if $j$ is not problematic, then we have the following: $cost(H_j) = A$, which implies that there exists $l \in \{ 1, \ldots, C_{j'} \}$ such that $L_j^l \subseteq D$ and such that the vertices of $K_j^l$ are only dominated by $L_j^l$ ; for any $i \in \{ 1, \ldots, n \}$, $cost(Q_{i,j}) = 2$, so exactly two vertices in $Q_{i,j}$ are in $D$, and these two vertices are not connected to the vertices of $K_j^l$ (since this set is only dominated by $L_j^l$). 

Consider now a non-problematic $j \in K$ and $i \in \{ 1, \ldots, n \}$. Since $cost(Q_{i,j}) = 2$, we claim that the solution must follow one of the six following configurations below:

(a) $u_{i,4j}, u_{i,4j+1} \in D$

(b) $u_{i,4j+1}, u_{i,4j+2} \in D$

(c) $u_{i,4j+2}, u_{i,4j+3} \in D$

(d) $u_{i,4j+3}, u_{i,4j} \in D$

(e) $u_{i,4j}, u_{i,4j+2} \in D$

(f) $u_{i,4j+1}, u_{i,4j+3} \in D$

Indeed, it is not hard to see that these six configurations cover all the cases where exactly two vertices of $Q_{i,j}$ are in $D$ (since $cost(Q_{i,j}) = 2$). 

\begin{claim}
There exists a contiguous interval $K^* \subseteq \{ 0, \ldots, Fm-1 \}$ of size at least $m$ in which all all $j \in K^*$ are not problematic and for all $j_1, j_2 \in K^*$, $Q_{i,j_1}$ and $Q_{i,j_2}$ are in the same configuration. 
\end{claim}

\begin{proof}
We make the following observations. For any $j \in K$ and any $i \in \{ 1, \ldots, n \}$, the vertices of $Q_{i,j}$ which are not in $D$ are only dominated by the vertices of the main part. Firstly, it is obvious if $x_i$ is not active in $c_{j'}$ (for $j' = j \mod m$) since in this case the vertices of $Q_{i,j}$ are not connected to any verification gadget. If $x_i$ is active in $c_{j'}$, it is also clear when we note that no vertex of $K_j$ is taken in the solution (since $cost(H_j) = A$). Moreover, the vertices of $Q_{i,j}$ which are in $D$ are not neighbors of vertices in $D$ outside the main part. It is again obvious if $x_i$ is not active in $c_{j'}$. If $x_i$ is active in $c_{j'}$, it is also clear since no vertex of $K_j$ is taken in $D$. Furthermore, the neighbors in the verification gadgets of the vertices of $Q_{i,j}$ not in the solution are all dominated by the vertices of $L_j^l$. From these observations, we obtain the following: the vertices of $Q_{i,j}$ which are in $D$ must have a private neighbor in the path $P_i$ or must be their own private vertex ; and the vertices of $Q_{i,j}$ which are not in $D$ must be dominated by the vertices in the path $P_i$. 

Now, given these observations, and the six configurations given before, we make the following statements, where a statement apply for any $i \in \{ 1, \ldots, n \}$ and $j$ such that $j$ and $j+1$ are in $K$:

\begin{itemize}
    \item If $Q_{i,j}$ is in configuration (a), then $Q_{i,j+1}$ is in configuration (a), (d) or (e)
    \item If $Q_{i,j}$ is in configuration (b), then $Q_{i,j+1}$ is in configuration (b) or (f)
    \item If $Q_{i,j}$ is in configuration (c), then $Q_{i,j+1}$ is in configuration (c)
    \item If $Q_{i,j}$ is in configuration (d), then $Q_{i,j+1}$ is in configuration (c), (d) or (f)
    \item If $Q_{i,j}$ is in configuration (e), then $Q_{i,j+1}$ is in configuration (b), (d), (e) or (f)
    \item If $Q_{i,j}$ is in configuration (f), then $Q_{i,j+1}$ is in configuration (c) or (f)
\end{itemize}

For the first statement, we have the following: (b), (c) and (f) cannot follow (a) since it would left at least one vertex not dominated. For the second statement, we have the following: (a), (d) and (e) cannot follow (b) since at least one vertex will not have a private neighbor ; (c) cannot follow (b) since it would left a vertex non dominated. For the third statement, we have the following: (a), (b), (d), (e) and (f) cannot follow (c) since at least one vertex will not have a private neighbor. For the fourth statement, we have the following: (a), (b) and (e)  cannot follow (d) since at least one vertex will not have a private neighbor. For the fifth statement, we have the following: (a) cannot follow (e) since at least one vertex will not have a private neighbor ; (c) cannot follow (e) since it would left a vertex non dominated. For the last statement, we have the following: (a), (b), (d) and (e) cannot follow (f) since at least one vertex will not have a private neighbor. 

For some $i \in \{ 1, \ldots, n \}$ and $j \in K$, we will say that $j$ is \textit{shifted} for variable $i$ if $j+1 \in K$ but $Q_{i,j}$ and $Q_{i,j}$ are not in the same configuration. We observe that there cannot exist distinct $j_1, j_2, j_3, j_4, j_5 \in K$ such that they are all shifted for variable $i$. Indeed, if we draw a directed graph with a vertex for each configuration and an arc $(u,v)$ expressing the property that the configuration represented by $v$ can follow the configuration represented by $u$, then we observe that the graph obtained is a DAG of maximum length 4. 

Then, by the above, the number of shifted indices $j \in K$ is at most $4n$. Hence, the longest contiguous interval without shifted indices has length at least $|K| / (4n+1) \geq m$, since $|K| \geq m(4n+1)$. Let $K^*$ be this interval. \qed
\end{proof}

We have located an interval $K^* \subseteq \{ 0, \ldots, Fm-1 \}$ of length at least $m$ where, for all $i \in \{ 1, \ldots, n \}$ and all $j_1, j_2 \in K^*$, we have the same configuration in $Q_{i,j_1}$ and $Q_{i,j_2}$. We now extract a satisfying assignment for $\varphi$ from this in the natural way. For some $j \in K^*$: if $Q_{i,j}$ is in configuration (a), then we set $x_i = 0$ ; if $Q_{i,j}$ is in configuration (b), then we set $x_i = 1$ ; if $Q_{i,j}$ is in configuration (c), then we set $x_i = 2$ ; if $Q_{i,j}$ is in configuration (d), then we set $x_i = 3$ ; if $Q_{i,j}$ is in configuration (e), then we set $x_i = 4$ ; if $Q_{i,j}$ is in configuration (f), then we set $x_i = 5$. We claim this satisfies $\varphi$. Consider a constraint $c_{j'}$ of $\varphi$. There must exist $j \in K^*$ such that $j' = j \mod m$ since $|K^*| \geq m$ and $K^*$ is contiguous. We therefore check the verification gadget $H_j$, where there exists $\sigma_l$ such that $L_j^l \subseteq D$ (this is because $j$ is not problematic, that is, $H_j$ attains its maximum cost). But because the vertices of $K_j^l$ are only dominated by the vertices $L_j^l$ and not by the vertices of the main part, it must be the case the the assignment we extracted agrees with $\sigma_l$, hence $c_{j'}$ is satisfied. This is true for all constraint $c_{j'}$ of $\varphi$. \qed
\end{proof}

We can now show that the pathwidth of $G$ is bounded by $n + O(1)$. 

\begin{lemma}\label{lem:pathwidthupperdom}
The pathwidth of $G$ is at most $n + O(1)$. 
\end{lemma}

\begin{proof}
We will show how to build a path decomposition of $G$. As in Lemma \ref{lem:barwardupperdom}, for all $j \in \{ 0, \ldots, Fm-1 \}$, let $V_j = H_j \cup \bigcup_{i=0}^n Q_{i,j}$, where $Q_{i,j} = \{ u_{i,4j}, u_{i,4j+1}$, $u_{i,4j+2}$, $u_{i,4j+3} \}$. We will show how to obtain a path decomposition of $G[V_j]$ with the following properties:

\begin{itemize}
    \item The first bag of the decomposition contains the vertices $u_{i,4j}$, for all $i \in \{ 1, \ldots, n \}$
    \item The last bag of the decomposition contains the vertices $u_{i,4j+3}$, for all $i \in \{ 1, \ldots, n \}$
    \item The width of the decomposition is $n + O(q6^q)$
\end{itemize}

We now show how to obtain such a decomposition of $G[V_j]$, having partially fixed the contents of the first and last bag of the decomposition. The verification gadget $H_j$ contains at most $2(4q+2)(6^q-1) + 1$ vertices (since $6^q-1$ is an upper bound on the number of assignments in the list of the corresponding constraint), so we place all its vertices in all bags. The remaining graph is a union of paths of length 4. We therefore have a sequence of $O(n)$ bags, where, for each $i \in \{ 1, \ldots, n \}$, we add to the current bag the vertices of $Q_{i,j}$ and then we add another bag with $Q_{i,j}$ removed except for $u_{i,4j+3}$. 

Now that we have found a path decomposition of $G[V_j]$ with the desired properties, we present how to obtain a path decomposition of the whole graph. The sets $V_j$ partition all remaining vertices of the graph (except the first three vertices and the last three vertices of each path $P_i$), while the only edges not covered by the above decompositions of $G[V_j]$ are those between the vertices $u_{i,4j+3}$ and $u_{i,4(j+1)}$. We therefore place the decompositions of $G[V_j]$ in order, and then, between the last bag of the decomposition of $G[V_j]$ and the first bag of the decomposition of $G[V_{j+1}]$, we have $2n$ "transition" bags, where in each transition step we add a vertex $u_{i,4(j+1)}$ in the bag, and then remove the corresponding vertex $u_{i,4j+3}$. 

We have now a path decomposition of the whole graph except the first three and the last three vertices of each path $P_i$, for all $i \in \{ 1, \ldots, n \}$. So, before the first bag of the decomposition of $G[V_0]$, we have a sequence of $O(n)$ bags, where, for each $i \in \{ 1, \ldots, n \}$, we add to the current bag the four vertices $u_{i,-3}, u_{i,-2}, u_{i,-1}$ and $u_{i,0}$ and then we add another bag with only the vertex $u_{i,0}$. We use the same method for the last three vertices of the paths $P_i$, after the decomposition of $G[V_{Fm-1}]$. 

Thus, we obtain a path decomposition of with $n+O(1)$. \qed
\end{proof}

We are now ready to present the main result of this section:

\begin{theorem}\label{th:pathwidthupperdom}
Under SETH, for all $\varepsilon > 0$, no algorithm solves \textsc{Upper Dominating Set} in time $O^*((6-\varepsilon)^{pw})$, where $pw$ is the input graph's pathwidth.
\end{theorem}

\begin{proof}
Fix $\varepsilon >0$ and let $q$ be sufficiently large so that Lemma \ref{lem:lampis} is true. Consider an instance $\varphi$ of \textsc{$q$-CSP-6}. Apply our reduction to obtain an instance $(G,k)$ of \textsc{Upper Domination}. Thanks to Lemmas \ref{lem:forwardupperdom} and \ref{lem:barwardupperdom}, we know that $\varphi$ is satisfiable if and only if there exists an upper dominating set of size at least $k$ in $G$. 

Now suppose that there exists an algorithm that solves \textsc{Upper Domination} in time $O^*((6-\varepsilon)^{pw})$. With this algorithm and our reduction, we can determine if $\varphi$ is satisfiable in time $O^*((6-\varepsilon)^{pw})$, where $pw = n + O(1) $ (Lemma \ref{lem:pathwidthupperdom}), so the total running time of this procedure is at most $O^*((6-\varepsilon)^n)$, contradicting SETH. \qed 
\end{proof}

\section{Sub-Exponential Approximation}\label{section:subexponential}

\subsection{Sub-Exponential Approximation Algorithm}

In this section, we present a sub-exponential approximation algorithm for the \textsc{Upper Dominating Set} problem. We prove the following: for any $r < n$, there exists an $r$-approximation algorithm for the \textsc{Upper Dominating Set} problem running in time $n^{O(n/r)}$.

To show this result, we use a common tool to design sub-exponential algorithms: partitioning the set of vertices $V(G)$ of the input graph into a convenient number of subsets of the same size. On each subset, we create a number of solutions: all maximal independent sets $I$ in the subgraph induced by the considered set of vertices ; and all subsets $S$ of the considered subset. For each maximal independent set $I$, we extend it to the whole graph. For each subset $S$, we first go through all subsets of neighbors of vertices of $S$ in order to find the correct set of private neighbors, and then we extend the solution to the whole graph. At the end, we output the best solution encountered. By computing all maximal independent sets $I$ and by going through all subsets $S$, we prove that there exists at least one valid upper dominating set which has the desired size. Note that, given a subset of an upper dominating set whose vertices have private neighbors, it may be impossible to extend the partial solution if we do not know their private vertices. This is why we need to find the private vertices of the subset $S$ we consider, since in our proof the solution which has the desired size may come from such a subset $S$. We prove the following:

\begin{theorem}\label{th:subexapproxuds}
For any $r < n$, \textsc{Upper Dominating Set} is $r$-approximable in time $n^{O(n/r)}$.
\end{theorem}

\begin{proof}
Let $D^* = S^* \cup I^*$ be any maximum upper dominating set of $G$, where $S^*$ is the set of vertices of $D^*$ which have some private neighbors, and $I^*$ is the set of vertices of $D^*$ which forms an independent set. 

We begin our algorithm by partitioning the set of vertices $V(G)$ into $l$ subsets $V_1, \ldots, V_l$, where $l = \lfloor\frac{r}{2} \rfloor$. 

Now, for each $i \in \{ 1, \ldots, l \}$, we do the following:
\begin{enumerate}
    \item Enumerate all maximal independent sets of $G[V_i]$. Let $\mathcal{I}_i$ be this family of independent sets.
    \item For each maximal independent set $I \in \mathcal{I}_i$, do the following:
    \begin{enumerate}
        \item Extend $I$ greedily to obtain an independent set $I'$ of the whole graph $G$, in the natural way: while there exists a vertex $u \in V \setminus N[I']$, add $u$ to $I'$. 
    \end{enumerate}
    \item Consider all subsets of vertices $S$ of $V_i$. 
    \item For each such subset $S \subseteq V_i$, do the following:
    \begin{enumerate}
        \item For each vertex $u \in S$, go trough all vertices $v \in N(u) \setminus N[S]$ so that the vertex $v$ is the private neighbor of $u$. 
        \item Let $P$ be the set of private neighbors of the vertices of $S$ found in the previous step, if such a set exists. 
        \item Let $N_{SP} = N(S) \cap N(P)$, $N_{S} = N(S) \setminus N_{SP}$, $N_P = N(P) \setminus N_{SP}$, $V_{SP} = V \setminus (N[S] \cup N[P])$, and $Q_P = N_P \setminus N(V_{SP})$. 
        \item We extend the partial solution $S$ as follows:
        \begin{enumerate}
            \item Let $T_1 = N(Q_P) \cap N_S$. 
            \item Greedily remove vertices of $T_1$ which have not a private neighbor in $N_P$, that is vertices $u \in T_1$ such that $(N(u) \cap N_P) \subseteq N(T_1 \setminus \{ u \})$. 
            \item Let $T_2 = N(N_P \setminus N(T_1)) \cap V_{SP}$. 
            \item Greedily remove vertices of $T_2$ which have not a private neighbor in $N_P \setminus N(T_1)$, that is vertices $u \in T_2$ such that $(N(u) \cap (N_P \setminus N(T_1))) \subseteq N(T_2 \setminus \{ u \})$. 
            \item Extend $S \cup T_1 \cup T_2$ greedily to obtain an upper dominating set $S'$ of the whole graph $G$, in the natural way: while there exists a vertex $u \in V_{SP} \setminus N[T_2]$, add $u$ to $S'$. 
            \item Discard $S'$ if it is not an upper dominating set of $G$. 
        \end{enumerate}
    \end{enumerate}
    \item Output the solution of maximum size encountered. 
\end{enumerate}

We first prove that our algorithm has the desired running-time. For each $i \in \{ 1, \ldots, l \}$, the set $V_i$ is of size roughly $\frac{n}{l} = 2n/r$, so we have that enumerating all maximaul independent sets of $G[V_i]$ takes time $O^*(3^{2n/3r})$, by the well-known result of Moon and Moser \cite{moon1965cliques} which states that computing all maximal independent sets of a graph of order $n$ can be done in time $O^*(3^{n/3})$. Moreover, by the same upper-bound on the size of the set $V_i$, we have that considering all subsets $S \subseteq V_i$ takes time $2^{2n/r}$, and there is that many subsets $S$. Now, observe that at the step 4.(a), for a vertex $u \in S$, we go through all vertices $v \in N(u) \setminus N[S]$, so through at most $n$ vertices, and that there is at most $2n/r$ such vertices $u \in S$. So for a subset $S \subseteq V_i$, we consider at most $n^{2n/r}$ sets $P$ of private neighbors of the vertices of $S$. Note that the other steps of our algorithm can be done in polynomial time. So the total running-time of our algorithm is:
\[ k \cdot (O^*(3^{2n/3r}) + 2^{2n/r} \cdot n^{2n/r}) = n^{O(n/r)} \]

Now, we will prove that our algorithm outputs an upper dominating set. Consider any $i \in \{ 1, \ldots, l \}$ and any maximal independent set $I \in \mathcal{I}_i$ of $G[V_i]$. Note that, since $I$ is a maximal independent set of $G[V_i]$, it can be easily extended to obtain a maximal independent set $I'$ of the whole graph $G$. Indeed, by greedily adding vertices of $V \setminus N[I']$, we obtain at the end of the step 2.(a) a maximal independent set of $G$, since every vertex of $V(G)$ is either in $I'$ or has a neighbor in $I'$. Note that, since $I'$ is maximal, it is also an upper dominating set: all vertices of $V(G)$ are dominated and the vertices of $I'$ form an independent set. So all the independent set $I'$ for all $i$ are valid upper dominating sets of the graph $G$. 

Consider any $i \in \{ 1, \ldots, l \}$. For the sets $S'$ constructed at step 4 of our algorithm, we will show that at least one of them is an upper dominating set of $G$. Since we consider all subsets $S$ of $V_i$, we consider the set $S_i^* = S^* \cap V_i$. Then, for each vertex $u$ in this set $S_i^*$, we consider
all its neighbors in $N(u) \setminus N[S_i^*]$ to be its private neighbor. So we consider the set $P_i^*$ which contains the private neighbor $v$ for each vertex $u \in S_i^*$ associated to the optimal solution $D^*$. Observe that the sets $N_{S_i^*P_i^*}$ and $N_{S_i^*}$ are dominated by $S_i^*$. Now consider the
vertices of the set $Q_{P_i^*}$: they are not neighbors of $V_{S_i^*P_i^*}$ by definition ; and they cannot be dominated by $N_{P_i^*} \cup N_{S_i^*P_i^*}$ since this set contains only neighbors of the vertices of $P_i^*$. So the vertices of $Q_{P_i^*}$ can only be dominated by vertices of $N_{S_i^*}$. By
our construction, the set $T_1$ is a set of vertices of $N_{S_i^*}$ which dominates $Q_{P_i^*}$ and such that each vertex $u \in T_1$ has a private neighbor. So the set $Q_{P_i^*}$ is dominated by $T_1$ and the vertices of $T_1$ each have at least one private neighbor (in $Q_{P_i^*}$ or in $N_{P_i^*} \setminus
Q_{P_i^*}$). Now consider the vertices of the set $N_{P_i^*} \setminus N(T_1)$: they cannot be in the solution since they are neighbors of $P_i^*$ ; and they all have at least one neighbor in $V_{S_i^*P_i^*}$ (since they were not in $Q_{P_i^*}$). By our construction, the set $T_2$ is a set of vertices of
$V_{S_i^*P_i^*}$ which dominates $N_{P_i^*} \setminus N(T_1)$ and such that each vertex $u \in T_2$ has a private neighbor in $N_{P_i^*} \setminus N(T_1)$: if a vertex of $T_2$ has no private neighbor in $N_{P_i^*} \setminus N(T_1)$, then it is removed from $T_2$ and $N_{P_i^*} \setminus N(T_1)$ stay dominated. Now observe that all vertices of $T_2$ have their private neighbor in
$N_{P_i^*} \setminus N(T_1)$. So we can greedily extend $S_i^*
\cup T_1 \cup T_2$ in a maximal independent set fashion by adding
vertices of $V_{S_i^*P_i^*} \setminus N[T_2]$ until the whole
graph $G$ becomes dominated. So the set ${S_i^*}'$ obtained is an
upper dominating set of $G$. So for any $i \in \{ 1, \ldots, l
\}$, there exists at least one set $S'$ which is an upper dominating set of $G$, and the non-valid solutions are discarded at the end of step 4.(d). 

Thus, the algorithm always outputs an upper dominating set.

Now, we will prove the approximation ratio. Note first that, since we have partitioned $V(G)$ into $l = \lfloor \frac{r}{2} \rfloor$ equal-size subsets $V_1, \ldots, V_l$, there exists $i^* \in \{ 1, \ldots, l \}$ such that $|D^* \cap V_{i^*}| \geq |D^*|/l \geq 2|D^*|/r$. Consider the corresponding subset $V_{i^*}$. Note now that, since $D^* = S^* \cup I^*$, we have the following: either at least $|D^* \cap V_{i^*}|/2$ vertices of $D^* \cap V_{i^*}$ are in $I^*$, or at least $|D^* \cap V_{i^*}|/2$ vertices of $D^* \cap V_{i^*}$ are in $S^*$. 

Suppose first that at least $|D^* \cap V_{i^*}|/2$ vertices of $D^* \cap V_{i^*}$ are in $I^*$. Since we have enumerating all maximal independent sets $I \in \mathcal{I}_{i^*}$ of $G[V_{i^*}]$, and since $I^* \cap V_{i^*}$ is an independent set of $V_{i^*}$, we have found at least one maximal independent set $I_{i^*}$ of $G[V_{i^*}]$ such that $I^* \cap V_{i^*} \subseteq I_{i^*}$. Then, we have extended $I_{i^*}$ to obtain a maximal independent set $I_{i^*}'$ of $G$. Thus, we have the following:
\[ |I_{i^*}'| \geq |I_{i^*}| \geq |I^* \cap V_{i^*}| \geq |D^* \cap V_{i^*}|/2 \geq 2|D^*|/2r = |D^*|/r \]

But since our algorithm outputs the maximum sized solution encountered, we have the desired approximation ratio in this case. 

Suppose now that at least $|D^* \cap V_{i^*}|/2$ vertices of $D^* \cap V_{i^*}$ are in $S^*$. Since we have considered all subsets $S$ of $V_{i^*}$, we have considered the subset $S_i^* = S^* \cap V_{i^*}$. To this set, we have considered all possible sets of private neighbors of vertices of $S_i^*$, and we have extended the set to an upper dominating set ${S_i^*}'$ of $G$ (note that the set $S_i^*$ has been successfully extended since it is the set we have considered when we have proved that at least one set $S'$ constructed at step 4 is a valid upper dominating set of $G$). Thus, we have the following:
\[ |{S_i^*}'| \geq |S_i^*| = |S^* \cap V_{i^*}| \geq |D^* \cap V_{i^*}| / 2 \geq 2|D^*|/2r = |D^*|/r \]

Again, since our algorithm outputs the maximum sized solution encountered, we have the desired approximation ratio in this case also. \qed
\end{proof}

\subsection{Sub-Exponential Inapproximability}

In this section, we give a lower bound on the complexity of any $r$-approximation algorithm, matching our algorithm of the previous section. We get the following result: for any $r < n$ and any $\varepsilon > 0$, there is no algorithm that outputs an $r$-approximation for the \textsc{Upper Dominating Set} problem running in time $n^{(n/r)^{1-\varepsilon}}$. 

To obtain this result, we will first prove the desired lower bound for the \textsc{Maximum Minimal Hitting Set} problem. In this problem, we are given an hypergraph and we want to find a set of vertices which cover all hyper-edges. Moreover, we need that this set is minimal, i.e. every vertex in the solution covers a private hyper-edge, and we want the solution to be of maximum size. 

To obtain this lower bound for the \textsc{Maximum Minimal Hitting Set} problem, we will do a reduction from the \textsc{Maximum Independent Set} problem. Then, we will make a reduction from the \textsc{Maximum Minimal Hitting Set} problem to the \textsc{Upper Dominating Set} problem to transfer this lower bound to our problem. 

Recall that we have the following lower bound by Chalermsook et al. \cite{ChalermsookLN13} for the \textsc{Maximum Independent Set} problem:

\begin{lemma}[Theorem 1.2 from \cite{ChalermsookLN13}]\label{th:chalermsook}
For any $\varepsilon > 0$ and any sufficiently large $r > 1$, if there exists an $r$-approximation algorithm for \textsc{Maximum Independent Set} running in time $2^{(n/r)^{1-\varepsilon}}$, then the randomized ETH is false. 
\end{lemma}

We note that making a reduction from the \textsc{Maximum Minimal Hitting Set} problem to derive hardness result for the \textsc{Upper Dominating Set} problem has already be done by Bazgan et al. \cite{BazganBCFJKLLMP18}. Indeed, to get the $n^{1-\varepsilon}$-inapproximability result for the \textsc{Upper Dominating Set} problem, they first derive this bound of the \textsc{Maximum Minimal Hitting Set} problem and then they designed an approximation-preserving reduction between these two problems, enabling them to transfer this hardness result to the \textsc{Upper Dominating Set} problem. 

In fact, to obtain the hardness result for the \textsc{Maximum Minimal Hitting Set} problem, they made a reduction from the \textsc{Maximum Independent Set} problem. Our first reduction is similar to this reduction and will allows us to get the desired hardness result for the \textsc{Maximum Minimal Hitting Set} problem. Our second reduction, from \textsc{Maximum Minimal Hitting Set} to \textsc{Upper Dominating Set} is the approximation-preserving reduction designed by Bazgan et al. \cite{BazganBCFJKLLMP18}. 

Note that our reduction from \textsc{Maximum Independent Set} to \textsc{Maximum Minimal Hitting Set} create a quadratic (in $n$) blow-up of the size of the instance of the latter problem. Such a blow-up does not allow us to derive the desired running-time. To answer this difficulty, we make another step in the reduction where we "sparsify" the instance of \textsc{Maximum Minimal Hitting Set} in order to keep the blow-up under control. To prove that the inapproximability gap stays the same, we  use a probabilistic analysis with Chernoff bounds. 

We will first prove the following hardness result:

\begin{theorem}\label{th:hardnesshittingset}
For any $\varepsilon > 0$ and any sufficiently large $r > 1$, if there exists an $r$-approximation algorithm for \textsc{Maximum Minimal Hitting Set} running in time $n^{(n/r)^{1-\varepsilon}}$, then the randomized ETH is false. 
\end{theorem}

\begin{proof}
First, we recall some details about the Lemma \ref{th:chalermsook}. To get this result, Chalermsook et al. \cite{ChalermsookLN13} made a reduction from an instance $\phi$ of \textsc{$3$-SAT} with $n$ variables, and for any $\varepsilon > 0$ and $r$ sufficiently large, they construct a graph $G$ with $|V(G)| = n^{1+\varepsilon}r^{1+\varepsilon}$ vertices which, with high probability, satisfies the following properties:

\begin{itemize}
    \item YES-instance: if $\phi$ is satisfiable, then $\alpha(G) \geq n^{1+\varepsilon}r$
    \item NO-instance: if $\phi$ is not satisfiable, then $\alpha(G) \leq n^{1+\varepsilon}r^{2\varepsilon}$.
\end{itemize}

Recall that $\alpha(G)$ is the size of a maximum independent set in $G$. 

With these properties, any approximation algorithm with ratio $r^{1-2\varepsilon}$ for \textsc{Maximum Independent Set} would distinguish whether $\phi$ is satisfiable or not, and so would solve the \textsc{$3$-SAT} instance. If this algorithm runs in time $2^{(n/r)^{1-\varepsilon}}$, then we obtain a sub-exponential algorithm for \textsc{$3$-SAT}, which contradicts the randomized ETH. 

Suppose that we are given $\varepsilon > 0$ and $r$ sufficiently large. Let $d = \frac{1}{\varepsilon^{1/2}}$ We will also design a reduction from the instance $\phi$ of \textsc{$3$-SAT} to an instance of \textsc{Maximum Minimal Hitting Set} going through an instance of \textsc{Maximum Independent Set} to show that an algorithm for the \textsc{Maximum Minimal Hitting Set} that achieves this ratio $r$ too rapidly would give a sub-exponential algorithm for \textsc{$3$-SAT}. So we start with the reduction of \cite{ChalermsookLN13}, from an instance $\phi$ of \textsc{$3$-SAT} on $n$ variables, and we adjust the parameter $r$ so that we obtain with high probability a graph $G$ with the following properties:

\begin{itemize}
    \item $|V(G)| = n^{1+\varepsilon}r^{1/d + \varepsilon / d}$
    \item YES-instance: if $\phi$ is satisfiable, then $\alpha(G) \geq n^{1+\varepsilon}r^{1/d}$.
    \item NO-instance: if $\phi$ is not satisfiable, then $\alpha(G) \leq n^{1+\varepsilon}r^{2\varepsilon /d}$. 
\end{itemize}

We now construct a graph $G'$ for the \textsc{Maximum Minimal Hitting Set} problem in the following way: we keep the graph $G$ ; for every subset $S \subseteq V(G)$ with $|S| = d$, we construct an independent set $Z_S$ of size $t = r^{1/d}$ ; and for every vertex $u \in Z_S$, we add the hyper-edge $S \cup \{ u \}$. Now, we claim that the graph $G'$ has the following properties:

\begin{itemize}
    \item $|V(G')| = \Theta(n^{d+d\varepsilon}r^{1+1/d + \varepsilon})$
    \item YES-instance: if $\phi$ is satisfiable, then $mmhs(G') = \Omega(n^{d+d\varepsilon}r^{1+1/d})$. 
    \item NO-instance: if $\phi$ is not satisfiable, then $mmhs(G') = O(n^{d+d\varepsilon}r^{1/d + 2\varepsilon})$.
\end{itemize}

Here, $mmhs(G')$ is the maximum size of a minimal hitting set in $G'$. 

Let us prove why the graph $G'$ has these properties. 

For the first property, note that there is $\binom{|V(G)|}{d}$ subsets $S$ of $V(G)$ of size $d$, and that for each of them we have added $t = r^{1/d}$ vertices in the corresponding independent set $Z_S$. So we have the following:

\[ |V(G')| = t \cdot \binom{|V(G)|}{d} + |V(G)| = \Theta(n^{d+d\varepsilon}r^{1+1/d+\varepsilon}) \]

For the second property, suppose that $\phi$ is satisfiable. It follows that $\alpha(G) \geq n^{1+\varepsilon}r^{1/d}$. We construct a minimal hitting set of $G'$ as follows: we take a minimum vertex cover $C$ of $G$ ; and for every subset $S$ of $I = V(G) \setminus C$ of size $d$, we take the $t$ vertices of the corresponding independent set $Z_S$. We observe that this solution is a minimal hitting set of $G'$. Indeed, $C$ is a minimum vertex cover of $G$, so all edges of $G$ are dominated by the solution, and every vertex of $C$ has at least one private edge since $C$ is a minimum vertex cover. Now observe that all the hyper-edges added in the construction of $G'$ which still have to be covered are hyper-edges between some vertices of the independent set $I = V(G) \setminus C$ and the corresponding independent sets $Z_S$, since all hyper-edges connected to the vertices of $C$ are covered. But we took the $t$ vertices of the independent set $Z_S$ of every subset $S \subseteq I$ of size $d$, so it follows that all the remaining hyper-edges are covered by our solution. Moreover, for any subset $S$ of $I$ of size $d$, note that $S$ is an independent set, so every vertex $u$ of $Z_S$ taken has a private hyper-edge, namely the hyper-edge $S \cup \{ u \}$. So our solution is a minimal hitting set. Now, let us determine its size. The number of independent sets $Z_S$ with $S \subseteq I$ of size $d$ is $\binom{\alpha(G)}{d}$. So the size of our solution is at least:
\[ t \cdot \binom{\alpha(G)}{d} = \Omega(n^{d+d\varepsilon}r^{1+1/d}) \]

For the third property, take any minimal hitting set in $G'$ and let $I$ be the corresponding independent set of $G$ ($I = V(G) \setminus C$ where $C$ is a vertex cover in $G$ which belongs to the minimal hitting set). We have that for any subset $S$ of $I$ of size $d$, the minimal hitting set takes at most the $t$ vertices of the independent set $Z_S$. And there is at most $\binom{\alpha(G)}{d}$ such subsets $S$. So the size of any minimal hitting set is bounded by:

\[ t \cdot \binom{\alpha(G)}{d} + |V(G)| = O(n^{d+d\varepsilon}r^{1/d+2\varepsilon}) \]

We have now construct a graph $G'$ of the \textsc{Maximum Minimal Hitting Set} problem where the gap between the values of $mmhs(G')$, corresponding on whether $\phi$ is satisfiable or not, is smaller than $r$ (it is $r^{1-2\varepsilon}$). Nonetheless, we cannot derive the desired hardness result since the order of $G'$ is quadratic on $n$. This blow-up makes it impossible to derive a sub-exponential algorithm for \textsc{$3$-SAT}. So we need to sparsify the gaph $G'$. 

Thus, we construct a graph $G''$ in the following way: we keep the graph $G'$ ; and we delete every vertex of $V(G') \setminus V(G)$ with probability $\frac{n^d-1}{n^d}$. That is, for every vertex $u$ in an independent set $Z_S$, the vertex $u$ stays in $G''$ with probability $\frac{1}{n^d}$. We claim that the graph $G''$ has the following properties:

\begin{itemize}
    \item $|V(G'')| = \Theta(n^{1+d\varepsilon}r^{1+1/d+\varepsilon})$
    \item YES-instance: if $\phi$ is satisfiable, then $mmhs(G'') = \Omega(n^{1+d\varepsilon}r^{1+1/d})$. 
    \item NO-instance: if $\phi$ is not satisfiable, then $mmhs(G'') = O(n^{1+d\varepsilon}r^{1/d +2\varepsilon})$. 
\end{itemize}

To establish these three properties, we will use the following Chernoff bound: suppose $X = \sum_{i=1}^p X_i$ is the sum of $p$ independent random 0/1 variables $X_i$ and that $E[X] = \sum_{i=1}^p E[X_i] = \mu$. We have the following: for all $0 \leq \delta \leq 1$, $Pr[|X-\mu| \geq \delta \mu] \leq 2e^{-\mu \delta^2/3}$. 

For the first property, we begin by defining a random variable $X_i$ for each vertex of each independent set $Z_S$ of $G'$: $X_i = 1$ if the corresponding vertex stays in $G''$ ; and $X_i = 0$ otherwise. Let $X$ be the sum of these $X_i$ variables, which is equal to the number of such vertices staying in $G''$. Suppose now that the number of vertices in the sets $Z_S$ in $G'$ is $cn^{d+d\varepsilon}r^{1+1/d+\varepsilon}$, where $c$ is a constant (it follows from the size of $V(G')$). Then $E[X] = cn^{1+d\varepsilon}r^{1+1/d+\varepsilon}$. We obtain $Pr[|X-E[X]| \geq \frac{E[X]}{2}] \leq 2e^{-E[X]/12} = o(1)$. So we conclude with high probability that $|V(G'')| = \Theta(n
^{1+d\varepsilon}r^{1+1/d+\varepsilon})$. 

For the second property, we consider a minimal hitting set $F$ of $G'$, of size $cn^{d+d\varepsilon}r^{1+1/d}$. We define a variable for each vertex of $F$ in the independent sets $Z_S$. As in the previous paragraph, we have that the expected number of such vertices which stay in $G''$ is $cn^{1+d\varepsilon}r^{1+1/d}$. Again, as in the previous paragraph, the actual number of such vertices will be close to this bound. We just need to prove that almost the same set is a minimal hitting set of $G''$. So we begin with the surviving vertices of $F$, which is an hitting set of $G''$ (since the removal of a vertex of $F$ implies the removal of its incident hyper-edges). Now, we delete vertices from $F$ until we obtain a minimal hitting set of $G''$. We will prove that the number of vertices deleted as redundant is at most $|V(G)| = n^{1+\varepsilon}r^{1/d+\varepsilon/d}$. Consider first an independent set $Z_S$ such that $Z_S \cap F \neq \emptyset$. Since $Z_S \cap F \neq \emptyset$, it follows that the vertices of the set $S$ are not in the solution $F$, because otherwise the vertices of $Z_S \cap F$ would not have a private hyper-edge. But because $S \cap F = \emptyset$, the vertices of $Z_S \cap F$ cannot be considered as redundant, since for every vertex $u \in Z_S \cap F$, it covers the hyper-edge $S \cup \{ u \}$. Thus, no vertex of the independent sets $Z_S$ can be removed as redundant. So the only vertices which can be removed as redundant are the vertices initially in $V(G)$. So at most $|V(G)| = n^{1+\varepsilon}r^{1/d+\varepsilon/d}$ vertices can be removed as redundant. Since $|V(G)| < \frac{c}{10}(n^{1+d\varepsilon}r^{1+1/d})$ (for $n$ and $r$ sufficiently large), it follows that removing these redundant vertices will not change the order of magnitude of the solution in $G''$. 

For the third property, we need to consider every possible minimal hitting set of $G''$ and prove that none of them is too large. So consider any subset $I \subseteq V(G)$ being an independent set of $G$. Our goal is to prove that any minimal hitting set $F$ of $G''$ that satisfies $V(G) \setminus I \subseteq F$ has a probability of being too big smaller than $2^{-|V(G)|}$. Indeed, if we prove this, we can take the union bound over all sets $I$ and conclude that with high probability no minimal hitting set of $G''$ is too big. So suppose now that we have fixed an independent set $I \subseteq V(G)$. We have $|I| \leq \alpha(G) \leq n^{1+\varepsilon}r^{2\varepsilon/d}$. We now make the following observation: any minimal hitting set $F$ which satisfies $V(G) \setminus I \subseteq F$ cannot contain any vertex of a set $Z_S$ if $S \cap F \neq \emptyset$ ; but may contain the $t$ vertices of an independent set $Z_S$ if $S \cap F = \emptyset$. The total number of such vertices in $G'$ is $O(n^{d+d\varepsilon}r^{1/d+2\varepsilon})$, since it is an upper bound on $mmhs(G')$. So, by the same argument as in the previous paragraph, the expected number of such vertices which stay in $G''$ is at most $\mu = cn^{1+d\varepsilon}r^{1/d+2\varepsilon}$, for a constant $c$. By using the Chernoff bound, we have $Pr[|X-\mu|\geq \frac{\mu}{2}] \leq 2e^{-\mu/12}$. We claim that $2e^{-\mu/12} = o(2^{-|V(G)|})$. Indeed, it follows since $|V(G)| = n^{1+\varepsilon}r^{1/d + \varepsilon/d} = o(\mu)$. Thus, the probability that a minimal hitting set being too large exists for a fixed independent set $I \subseteq V(G)$ is low enough so that taking the union bound over all possible independent sets $I$ give a probability that at least one minimal hitting set is too big of value $o(1)$. So we have with high probability that no minimal hitting set of size greater than $3\mu /2$ exists. So we obtain the third property. 

Now that we have proved that $G''$ satisfies these three properties, we will show how to obtain the Theorem. Suppose that, for sufficiently large $r$ and any $\varepsilon > 0$, there exists an approximation algorithm for \textsc{Maximum Minimal Hitting Set} with ratio $r^{1-3\varepsilon}$ running in time $N^{(N/r)^{1-4\varepsilon}}$ for graphs of order $N$. The ratio of this algorithm is sufficiently small to distinguish between the two cases in our graph $G''$, as the ratio between $mmhs(G'')$ when $\phi$ is satisfiable or not is $\Omega(r^{1-2\varepsilon})$ (for $r$ sufficiently large). So we can use this approximation algorithm to solve \textsc{$3$-SAT}. Furthermore, we have the following:
\[ N/r = \Theta((n^{1+d\varepsilon}r^{1+1/d+\varepsilon})/r) = O(n^{1+(d+1)\varepsilon+1/d}) = O(n^{1+\varepsilon + 2\varepsilon^{1/2}}) \]

Therefore, $(N/r)^{1-4\varepsilon} = o(n)$. We obtain an algorithm for \textsc{$3$-SAT} in time $N^{(N/r)^{1-\varepsilon}} = 2^{n^{1-\varepsilon'}}$ for $\varepsilon' < \varepsilon$ chosen appropriately. This contradicts the randomized ETH. So by adjusting $r$ and $\varepsilon$, we get that no $r$-approximation algorithm for \textsc{Maximum Minimal Hitting Set} can run in time $N^{(N/r)^{1-\varepsilon}}$ for graphs of order $N$. Thus we get the statement of the Theorem. \qed
\end{proof}

With this hardness result for \textsc{Maximum Minimal Hitting Set}, and by using the reduction of Bazgan et al. \cite{BazganBCFJKLLMP18}, we get the following hardness result for \textsc{Upper Dominating Set}:

\begin{theorem}\label{th:hardnesssubexponentialuds}
For any $\varepsilon > 0$ and any sufficiently large $r > 1$, if there exists an $r$-approximation for \textsc{Upper Dominating Set} running in time $n^{(n/r)^{1-\varepsilon}}$, then the randomized ETH is false. 
\end{theorem}

\begin{proof}
We start with an instance $G$ of \textsc{Maximum Minimal Hitting Set} obtained from Theorem \ref{th:hardnesshittingset}. From this instance, we construct an instance $G'$ of \textsc{Upper Dominating Set}. By Theorem 12 of \cite{BazganBCFJKLLMP18}, we know that this reduction is approximation-preserving. So the gap from the hardness result of \textsc{Maximum Minimal Hitting Set} stays the same for \textsc{Upper Dominating Set}. Now observe that in Theorem \ref{th:hardnesshittingset}, the number of hyper-edges in $G$ has the same order of magnitude than the number of vertices in $G$. Thus, the number of vertices in $G'$ is linearly dependent on the number of vertices in $G$. So an $r$-approximation algorithm for \textsc{Upper Dominating Set} running in time $n
^{(n/r)^{1-\varepsilon}}$ would give an $r$-approximation algorithm for \textsc{Maximum Minimal Hitting Set} with the same running-time, which would contradicts Theorem \ref{th:hardnesshittingset} and the randomized ETH. So we obtain the desired hardness result for \textsc{Upper Dominating Set}. \qed
\end{proof}

%
%
%
\bibliographystyle{splncs04}
\bibliography{samplepaper}
%

%
%

\end{document}